\documentclass[a4paper,11pt]{article}

\usepackage[english]{babel}
\usepackage[utf8]{inputenc}
\usepackage[T1]{fontenc}

\usepackage{graphicx}
\usepackage[margin=1in]{geometry}
\usepackage{algorithm}
\usepackage[noend]{algpseudocode}
\usepackage{amssymb,amsmath}
\usepackage{amsthm}
\usepackage{thmtools}  
\usepackage{mathtools}
\usepackage{bm}
\usepackage{xcolor}
\definecolor{VUBlue}    {rgb}{0,    0.466, 0.702}
\definecolor{VUBlue2}   {rgb}{0.364,0.678, 0.815}
\definecolor{VUGreen}   {rgb}{0,    0.502, 0.325}
\definecolor{VUGreen2}  {rgb}{0.309,0.686, 0.282}
\definecolor{VUOrange}  {rgb}{0.8,  0.254, 0}
\definecolor{VUOrange2} {rgb}{0.909,0.411, 0.176}
\definecolor{VUPurple}  {rgb}{0.231,0.129, 0.443}
\definecolor{VUPurple2} {rgb}{0.556,0.302, 0.643}
\usepackage{csquotes}
\usepackage{nicefrac}
\usepackage[
    backend=biber,
    style=alphabetic,
    sorting=ynt
]{biblatex}
\usepackage[colorlinks=true, allcolors=VUBlue]{hyperref}
\usepackage[capitalize]{cleveref}
\usepackage{placeins}
\usepackage{subcaption}
\usepackage{booktabs}
\usepackage{microtype}

\usepackage{newtxtext}
\usepackage{newtxmath}

\usepackage{enumitem}
\setlist[1]{labelindent=\parindent}
\setlist[enumerate]{label=(\arabic*)}
\setlist[itemize]{noitemsep}
\setlist[description]{noitemsep}

\usepackage{tikz}
\usepackage{pgfplots}
\usetikzlibrary{arrows.meta,positioning}
\tikzset{
  vu vertex/.style={circle,draw=black,fill=white,text=black,line width=.7pt,minimum size=6mm,inner sep=1pt,font=\small},
  vu arc/.style={-{Stealth[length=2mm]},draw=black,line width=.9pt},
  vu pair/.style={draw=black!55,line width=.8pt},
  vu fractional/.style={{Stealth[length=1.8mm]}-{Stealth[length=1.8mm]},draw=VUPurple,line width=.8pt,densely dashed},
  vu panel/.style={draw=black!50,fill=black!3,rounded corners=3pt,line width=.6pt},
  vu title/.style={font=\small\bfseries,text=black},
  vu text/.style={font=\small,align=center,text=black},
  vu box/.style={draw=black!65,fill=white,text=black,rounded corners=2pt,align=center,font=\small,inner sep=5pt}
}

\usepackage{authblk}
\usepackage{orcidlink}

    \author[1,2]{Ben Bals \orcidlink{0009-0009-1054-8444}}
    \author[2]{Matei Tinca \orcidlink{0009-0009-8109-1867}}
    \author[2,1]{Yasamin Nazari \orcidlink{0000-0003-1315-9355}}
    \affil[1]{CWI, Amsterdam, The Netherlands}
    \affil[2]{Vrije Universiteit, Amsterdam, The Netherlands}
    \date{November 4, 2026}

\declaretheorem[numberwithin=section]{theorem}

\declaretheorem[sibling=theorem]{corollary}
\declaretheorem[sibling=theorem]{lemma}
\declaretheorem[sibling=theorem]{definition}

\declaretheorem[sibling=theorem]{remark}

\declaretheorem[sibling=theorem]{observation}
\declaretheorem[sibling=theorem]{proposition}

\makeatletter
\def\theHALG@line{\theHalgorithm.\arabic{ALG@line}}
\makeatother

\usepackage{mathtools}
\usepackage{xspace}
\usepackage{tikz}
\usetikzlibrary{decorations.pathreplacing}

\newcommand{\cO}{\mathcal{O}}
\newcommand{\cOt}{\widetilde{\mathcal{O}}}
\newcommand{\Oish}{\cOt}
\newcommand{\smol}{\varepsilon}

\newcommand{\pred}{\mathrm{pred}}
\newcommand{\predvec}{\bm{\pred}}
\newcommand{\setsize}[1]{\left|#1\right|}
\newcommand{\cut}{\operatorname{cut}}
\newcommand{\cross}{\mathrm{cross}}
\newcommand{\out}{\operatorname{out}}
\newcommand{\xvec}{\bm{x}}
\newcommand{\yvec}{\bm{y}}
\newcommand{\zvec}{\bm{z}}
\newcommand{\hvec}{\bm{h}}
\newcommand{\dvec}{\bm{\outdegree}}
\newcommand{\evec}{\bm{e}}

\newcommand{\Dvec}{\bm{D}}

\newcommand{\zerovec}{\bm{0}}
\newcommand{\outvec}{\bm{\mathrm{out}}}

\newcommand{\pos}[1]{(#1)_+}

\newcommand{\norm}[1]{\left\lVert#1\right\rVert}

\newcommand{\ie}[1]{(i.e.,~#1)}
\newcommand{\eg}[1]{(e.g.,~#1)}

\title{Reachability with Subquadratic Cut Queries on Directed Graphs}

\begin{document}

\maketitle

\begin{abstract}
  In the \emph{cut-query model}, we have access to a (directed) graph via an oracle and we can  query the size of the (directed) cut of a given subset of the vertices.
  One of the most elementary tasks in this model is to decide if there is a path between two fixed vertices $s$ and $t$.

  While many results are known for undirected graphs, much less is understood for directed graphs in the cut-query model.
  Even for the basic task of $s$-$t$ reachability, the best known \emph{randomized} algorithm is to reconstruct the entire graph with a technique by Grebinski and Kucherov using $\cO(n^2 / \log n)$ queries \cite{grebinskiOptimalReconstructionGraphs2000a}.

  On \emph{directed acyclic graphs (DAGs)}, we obtain a \emph{deterministic} single-source reachability algorithm using $\cO(n \sqrt{n \log n})$ queries.
  The result is based on a topological sort algorithm and can also be adapted to compute single-source shortest paths in DAGs.

  On general directed graphs, we obtain \emph{deterministic} algorithms for single-source reachability, strongly connected components and minimum $s$-$t$ cut using $\cO(n^{9/5}\log n)$ queries.
  The algorithm is simple and relies on two natural greedy subroutines and a modified version of Dinic's algorithm.
\end{abstract}

\hypersetup{pageanchor=false}
\thispagestyle{empty}
\clearpage
\pagenumbering{arabic}
\setcounter{page}{1}
\hypersetup{pageanchor=true}

\section{Introduction}
We focus on the \emph{cut-query model}, where the goal is to perform a task on an input graph $G=(V,E)$ that we can only access via an oracle, with minimal number of queries.
In particular, given a subset of the vertices $S$, the oracle will return the value of the cut \ie{the edges from $S$ to $V \setminus S$}.  We write $\cut(S) \coloneqq \setsize{E \cap S \times (V \setminus S)}$ for such a query.
The first natural question is how to recover the entire graph, but it can be shown through an information theoretic argument that this requires nearly quadratically many queries.
Thus, our goal is to see which specific questions about the graph can be answered using fewer than quadratic queries.
Examples include: Can node $s$ reach node $t$? What is the minimum $s$-$t$ cut? What is the global min-cut? 

The cut-query model was introduced in 2018 by by Rubinstein, Schramm and Weinberg \cite{rubinstein_et_al:LIPIcs.ITCS.2018.39}.
The first time a similar model was studied was in 2000 when Grebinski and Kucherov \cite{grebinskiOptimalReconstructionGraphs2000a} showed how to fully recover an undirected graph using queries of the form: Given $S \subseteq V$, how many edges lie inside $S \times S$?
They show how to achieve full graph recovery using $\cO(n^2 / \log n)$ queries, which is optimal under an information theoretic lower bound. Their technique can be adapted to the cut-query model leading to the same number of queries.

One motivation of this model comes from submodular function optimization: the cut function is submodular and so the cut-query complexity of a problem also determines how many times this submodular function has to be queried.
Minimizing this is natural when computing the cut itself \ie{when we look for a minimum $s$-$t$ cut or a global minimum cut}.

The second motivation is that the cut-query model can be seen as a weaker version of the communication complexity of a problem. In that setting, two players receive an arbitrary disjoint part of the input graph each and an the goal is to answer  a question about the graph \eg{$s$-$t$ reachability} while minimizing the number of bits of communication \cite{DBLP:conf/stoc/HajnalMT88,DBLP:conf/focs/BlikstadBEMN22}.
Formally, a simple argument shows that any communication complexity lower bound translates to a cut-query lower bound and that any cut-query algorithm translates to a communication protocol (up to $\cO(\log n)$ bits) \cite{DBLP:phd/ndltd/Harvey08,DBLP:journals/corr/abs-2109-02115}.
We currently know of no natural\footnote{There are degenerate examples. As noted below, the cut-query model can only solve questions whose answer does not change when we flip the direction of all edges on any directed cycle. This gives the following simple example: \emph{Given $u,v$, is there an edge $u \to v$?} It can be solved using $\Oish(1)$ bits of communication, but cannot be solved at all using cut-queries.} graph problem whose cut-query and communication complexities differ by a polynomial.
As there is an $\Oish(n)$ cost\footnote{The $\cOt$ notation suppresses polylogarithmic factors.} communication protocol for $s$-$t$ reachability, it is currently a candidate.

In the first work on the cut-query model \cite{rubinstein_et_al:LIPIcs.ITCS.2018.39}, the authors present algorithms, for undirected graphs, for \textit{exact minimum global cut} and \textit{exact minimum $s$-$t$ cut} in complexities $\cOt(n)$ and $\cOt(n^{5/3})$ respectively, avoiding the bottleneck of rebuilding the whole graph in almost quadratically many queries.
Very recently, the complexity for \textit{minimum $s$-$t$ cuts} was further improved by Jiang, Nanongkai and Sawettamalya \cite{DBLP:conf/soda/JiangNS26} to $\cOt(n^{8/5})$.
Other problems (like spanning forests \cite{DBLP:conf/alt/0001C24}), or other models (like minimizing the number of rounds, each of arbitrarily many queries, \cite{DBLP:conf/esa/Kenneth-Mordoch25} or hypergraphs, \cite{DBLP:conf/esa/ChakrabartyL26}) have received attention too.

The problem is less understood for directed graphs. Here, while the oracle allows us to discover the endpoints of edges in a cut, we do not immediately also know about their direction. Specifically, as stated in \cite{graurNewQueryLower2020}, flipping the directions of any cycle would not change the query profile of a graph, hence the directions of edges inside strongly connected components cannot be deduced. Moreover, even the problem in which we have to decide whether the minimum cut has nonzero value is non-trivial: we know no algorithm better than $\cOt(n^2)$, nor a lower bound stronger than $\Omega(n)$.

For directed graphs, there is a trivial way to solve reachability by reconstructing all the edges without their direction and by querying the out-degrees of every vertex, using $\binom{n}{2} + n$ queries. Any subsequent cut can be computed with the available information without invoking the oracle.
This is based on the observation that we can compute the value of a cut $S$ by adding up all the degrees of nodes in $S$ and subtracting the number of internal edges of $S$, and these can be computed via standard subroutines in the cut-query model.

A very natural question to ask if $s$-$t$ reachability (or more generally single-source) can be solved in subquadratic queries on directed acyclic graphs (DAGs) as a first step to solving the general problem. Our first result is answering this question with a simple algorithm. Our second result is a more involved algorithm for general directed graphs.

\subsection{Our Contributions}

We make contributions both on DAGs and on general directed graphs. They use different techniques and we list them separately here.

\paragraph{Directed Acyclic Graphs.}
For DAGs, we present an algorithm that decides \textit{$s$-$t$ reachability} using $\cO\left( n \sqrt{n \log n}  \right)$ cut queries. At its core, the algorithm employs a procedure that computes a topological sort of a directed acyclic graph, also using $\cO(n \sqrt{n \log n})$ queries. Formally, we prove the following theorem.

\begin{restatable*}{theorem}{algtopsort} \label{thm:algtopsort}
    Let $G = (V, E)$ be a simple directed acyclic graph. There is an algorithm that computes a topological sort of $G$ using $\cO(n \sqrt{n \log n})$ cut queries.
\end{restatable*}

From this core result, we deduce the following. 

\begin{restatable*}{theorem}{correachability}\label{thm:reachability}
    Let $G = (V, E)$ be a directed acyclic graph and $s \in V$. There is an algorithm that computes all vertices that are reachable from $s$ using $\cO(n \sqrt{n \log n})$ cut queries.
\end{restatable*}

We also show this algorithm can be extended to compute distances.

\begin{restatable*}{theorem}{algsssp} \label{thm:sssp}
    Let $G = (V, E)$ be a directed acyclic graph and $s\in V$. There is an algorithm that computes all the shortest paths from $s$ using $\cO(n \sqrt{n \log n})$ cut queries.
\end{restatable*}

\paragraph{General Directed Graphs.}

For general directed graphs, we present deterministic algorithms for single-source reachability, strongly connected components, and minimum $s$-$t$ cut using $\cO(n^{9/5}\log n)$ cut queries. At their core, these algorithms construct a \emph{pseudo-orientation}, which specifies how to orient the edges of the underlying undirected multigraph, called the \emph{skeleton}, without reconstructing it. The induced orientation has the same out-degrees as $G$. Formally, we prove the following theorem.

\begin{restatable*}{theorem}{alggeneralorientation}\label{thm:pseudo-orientation}
    Let $G = (V, E)$ be a simple directed graph. There is a deterministic algorithm that computes a pseudo-orientation $\yvec$ inducing an integral orientation of the skeleton of $G$ with the same out-degrees as $G$, using $\cO(n^{9/5}\log n)$ cut queries.
\end{restatable*}

The induced orientation has the same directed cut function as $G$, and hence the same reachability relation. Once the pseudo-orientation is known, we can find all vertices reachable from a given source using only $\cO(n\log n)$ additional queries.
That is, we show the following.

\begin{restatable*}{theorem}{alggeneralreachability}\label{thm:general-reachability}
  Let $G = (V, E)$ be a simple directed graph and $s \in V$. There is a deterministic algorithm that, given a pseudo-orientation with the same out-degrees as $G$,  computes all vertices that are reachable from $s$ in $\cO(n \log n)$ additional cut queries.
\end{restatable*}

By combining forward and backward graph searches, we can also compute the strongly connected components in the same number of queries.

\begin{restatable*}{theorem}{alggeneralscc}\label{thm:det}
    Let $G = (V, E)$ be a simple directed graph. There is a deterministic algorithm that, given a pseudo-orientation with the same out-degrees as $G$, computes the strongly connected components of $G$, in topological order, using $\cO(n \log n)$ additional cut queries.
\end{restatable*}

Finally, we show that the path-reversal procedure used to construct the pseudo-orientation can also be used to compute a minimum $s$-$t$ cut. Since the induced orientation preserves all directed cuts, it also preserves minimum cut values and, by the max-flow min-cut theorem, maximum flow values.
Note that this algorithm only works for unit capacities.

\begin{restatable*}{theorem}{alggeneralmincut}\label{thm:min-cut-flow}
    Let $G = (V, E)$ be a simple directed graph and $s,t \in V$ distinct vertices. There is a deterministic algorithm that, given a known pseudo-orientation with the same out-degrees as $G$, computes a minimum $s$-$t$ cut in $G$ and the maximum $s$-$t$ flow value with unit edge capacities using $\cO(n^{5/3}\log n)$ additional cut queries.
\end{restatable*}

\paragraph{Note on Independent Work.} 
We were informed that two other groups concurrently and independently obtained improved bounds for reachability on DAGs in the cut-query model. These works were published as preprints at the same time. In particular, Sanjeev Khanna, Aaron Putterman, and Junkai Song obtained an algorithm with near-linear query complexity \cite{khanna2026reachability}. Deeparnab Chakrabarty and Andrew Zhao obtained the same bounds as ours.

We discuss an extension to computing shortest paths on DAGs (\Cref{sec:dag-shortest-paths}) and how to topologically sort SCCs (\Cref{sec:toposort-scc}) not present in the concurrent works, but their algorithm (with the better query complexity) can also be adapted with the same techniques.

For the general case, in August 2026 an AI-generated manuscript claiming an $\cOt(n^{3/2})$ query bound \cite{directedReachabilityCutQueries} was posted.
This manuscript has no authors, has not been verified, and relies on complex continuous optimization and spectral sparsification tools.
Our goal was to present a simpler combinatorial solution, which we made an effort to make understandable. 
We were also just informed of another concurrent work with subquadratic bounds for the general graph that has been submitted to STOC 2027 \cite{fleischmannMukhopadhyaySanyal2026}.  They use similar techniques as ours (also with assistance from AI), but our approach is deterministic, while they use randomized rounding step.

\subsection{Organization and Technical Overview}
\Cref{section:preliminaries} presents the notation used throughout this document. In \Cref{section:subroutines}, we derive important subroutines. The identities used to compute helper functions (such as the degree) using cut queries are given succinctly in the proof of \Cref{lemma:primitives}. Then, \Cref{lemma:discover-edge} describes a procedure used to discover an edge between two disjoint sets of vertices $A$ and $B$ using $\cO(\log n)$ queries. \Cref{lemma:discover-edge} can be further used to reconstruct the graph ignoring the directions of each edge, similarly to \cite[Section 2]{rubinstein_et_al:LIPIcs.ITCS.2018.39}.

\paragraph{DAGs.} \Cref{section:topsort} adapts an algorithm that computes a topological sort from the classical model to the cut-query model. Then, the algorithm is modified to use a subquadratic number of queries by classifying vertices as \textit{heavy} or \textit{light} and processing them accordingly.
We construct the topological sort iteratively and always maintain a prefix of a topological sort to which we append.
The core insight is that a vertex that has many incoming edges from vertices not yet placed in our partial topological sort can be effectively ignored for a while.
In particular, we prove \Cref{thm:algtopsort}.
\Cref{section:applications} shows how the single-source reachability problem can be solved after computing the topological sort. The algorithm passes through each vertex once and decides for each vertex whether it is reachable or not. Interestingly, the pass itself uses only a linear number of queries, therefore the bottleneck for reachability is computing the topological sort.
We also show how a simple extension of our algorithm can lead to a single-source shortest path algorithm.

\paragraph{General Graphs.}
Imagine we could find all edges without their direction (the so-called \emph{skeleton}).
Then we could, without further cut queries, calculate an orientation of them that has the same out-degrees as the true graph $G$.
An easy argument shows that the reachability under this orientation must be the same as in $G$ and we could answer any reachability questions without any further cut queries.

Unfortunately, a simple counter-example shows that this scheme cannot work in polynomially subquadratic queries.
But, it turns out, we can slightly relax the notion of an orientation to circumvent this limitation.
We do not aim to find the skeleton and thus also not an orientation.
Instead, we find a \emph{pseudo-orientation}, which, for each pair of vertices $u,v$ specifies: \emph{If} there were edge(s) between $u$ and $v$, in which direction would each of them point?

Somewhat surprisingly, we can calculate the degree of a vertex \emph{under a known pseudo-orientation} using $\cO(1)$ cut queries.
Similar to orientations, if the out-degrees of a pseudo-orientation agree with those in the true $G$, the reachability (and indeed all cut-values) under it are equal to the true hidden graph $G$.
Once we know a pseudo-orientation with this property, there is a straightforward way to implement graph searches \ie{DFS or BFS} in $\cO(n \log n)$ additional cut queries.
These ideas and the resulting plan are discussed in \Cref{sec:orientations}.

So the question becomes: how can we find a pseudo-orientation with the correct out-degrees?
We do this in two steps:
First, in \Cref{sec:det}, we greedily construct a pseudo-orientation whose $\ell_1$ out-degree errors are bounded.
Second, in \Cref{sec:repair}, we run a modified version of Dinic's max-flow algorithm to repair these degree errors.
Dinic's algorithm finds augmenting paths that---in a correctly defined auxiliary graph---correspond to paths between a vertex with too high out-degree under the pseudo-orientation to one with too little.
Flipping the direction of all edges on the path in the pseudo-orientation reduces the $\ell_1$ out-degree error by two.
Since we have a bounded number of errors, we can bound the number of phases and the number of explored edges in Dinic's algorithm.
Correctly balancing the parameters used in these routines then gives the claimed bound of $\cO(n^{9/5} \log n)$ cut queries.
Finally, in \Cref{sec:consequences}, we deduce the reachability, SCC, and cut/flow algorithms.

\section{Preliminaries} \label{section:preliminaries}

For a real number $a$, we write $\pos{a}:=\max\{a,0\}$ for its positive part.
We denote vectors by bold symbols and their individual entries by ordinary symbols: for example, $\xvec$ has entries $x_{uv}$. 

Let $G = (V, E)$ be a simple directed graph. We define the number of vertices and the number of edges to be $n := |V|$ and $m := |E|$ respectively. The set of edges $E$ is a set of ordered pairs $E \subseteq V \times V$. An edge defined by the pair $(u, v) \in E$ will be denoted as $u \to v$.

We define the undirected graph $\tilde{G} = (V, \tilde{E})$ as the graph obtained from $G$ by making every edge undirected. That is, $\tilde{E}$ is instead a set of unordered pairs over $V \times V$. Then, an edge of $\tilde{E}$ is denoted by $uv$ and $vu$. We call $\tilde{G}$ the \textit{skeleton} of $G$.

A (directed) path from $u$ to $v$ is a sequence $u = s_1, s_2, \ldots, s_k = v$ of distinct vertices such that for all $i$ with $1 \le i \le k - 1$, there is an edge $s_i \to s_{i + 1}$. We denote that there exists a path from $u$ to $v$ as $u \rightsquigarrow v$, in which case we say that $v$ is reachable from $u$. Similarly, a (directed) cycle is a sequence $s_1, s_2, \ldots, s_k$ with $k > 2$ such that $s_1 = s_k$, all $s_1, \ldots, s_{k - 1}$ are distinct, and for all $i$ with $1 \le i \le k - 1$, we have $s_i \to s_{i + 1}$.

A \textit{directed acyclic graph} (DAG for short) is a directed graph which contains no cycles. For DAGs, informally, we can order the vertices in a sequence such that all edges go from left to right. This topological ordering is defined formally as the following:

\begin{definition} \label{def:topsort}
    Let $G = (V, E)$ be a directed acyclic graph. A  permutation $S = [s_1, s_2, \ldots, s_n]$ of vertices is a \textit{topological ordering} (also called a \textit{topological sort}) if for every edge $s_i \to s_j$, we have $i < j$.
\end{definition}

For a sequence of elements $S = [s_1, s_2, \ldots, s_n]$, it is sometimes helpful to argue on one of its prefixes. In this case we use $S_k := [s_1, s_2, \ldots, s_k]$ to denote the prefix of length $k$ of $S$.

A strongly connected component (SCC for short) is a maximal set of vertices $C \subseteq V$ such that, for every pair of vertices $u, v \in C$, we have $u \rightsquigarrow v$. For a vertex $v$, we denote $SCC(v)$ to be the strongly connected component which contains $v$. The vertices of a graph $G$ can then be partitioned into $k$ strongly connected components as $\mathcal{C} = \{C_1, C_2, \ldots, C_k\}$ with all $C_i$ disjoint sets and $\bigcup_{C \in \mathcal{C}} C = V$. Note that such a partition is unique.

Contracting the graph by its strongly connected components results in a DAG. As such, we can topologically order the SCCs of a graph.

\begin{definition}
    Let $G = (V, E)$ be a directed graph. Let $\mathcal{C} = \{C_1, C_2, \ldots, C_k\}$ be a partition of $V$ into strongly connected components. A permutation $\mathcal{S} = [S_1, S_2, \ldots, S_k]$ of $\mathcal{C}$ is a topological ordering of SCCs if for every edge $u \to v$ with $u \in S_i$, $v \in S_j$, and $i \neq j$, we have that $i < j$.
\end{definition}

\paragraph*{Cut Query Oracle}

\newcommand{\query}{\mathrm{cut}}
\newcommand{\querycomp}{\query_I}
\newcommand{\indegree}{\mathrm{d}^{-}}
\newcommand{\remindegree}{\indegree_{\mathrm{rem}}}
\newcommand{\outdegree}{\mathrm{d}^{+}}
\newcommand{\internal}{\mathrm{int}}
\providecommand{\cross}{\mathrm{cross}}

We study algorithms that do not have access to the structure of the graph, but instead they can access an oracle that returns the size of a cut. Concretely, an algorithm has access to $V$ and to the function $\query : \mathcal{P}(V) \to \mathbb{N}$ defined as

\[
    \query(S) := |\{u \to v : u \in S, v \not\in S\}|.
\] 

\noindent The objective of one such algorithm is to compute some information using as few oracle queries as possible. It may happen that an algorithm takes arbitrary running time.

\section{Useful Subroutines} \label{section:subroutines}

If a simple query counts the number of edges going in one direction \ie{outgoing edges}, we can compute the number of incoming edges into a set of vertices by applying a query on the complement:

\begin{align*}
    \querycomp(S) &:= |\{u \to v : u \not\in S, v \in S\}| \\
                  &~= \query(V \setminus S).
\end{align*}

Two simple subroutines that will prove to be useful later on are the degrees of each vertex:

\begin{align*}
    \outdegree(v) &:= \query(\{v\}), \\
    \indegree(v) &:= \querycomp(\{v\}).
\end{align*}

We write $\dvec=(\outdegree(u))_{u\in V}$ for the out-degree vector.

There is a simple way to decide whether there is an edge between two vertices $u, v$. If there is no edge in between them, we expect that the size of the cut $\{u, v\}$ is exactly the sum of their outdegrees. On the other hand, if there are edges in between the two, the sum of outdegrees would overcount said extra edges, thus the sum of outdegrees would be larger by one or two than the size of the cut.

To generalize this observation for larger subsets of vertices, we define the number of edges with both endpoints in the set $S$ as

\[
    \internal(S) := |\{u \to v : u \in S, v \in S\}|.
\] 

\noindent Then, by analyzing the difference between the sum of outdegrees and a cut query, we obtain:

\begin{align*}
    \left( \sum_{u \in S} \outdegree(u) \right) - \query(S) &= |\{u \to v : u \in S, v \in V\}| - |\{u \to v : u \in S, v \in V \setminus S\}| \\
                                                            &= |\{u \to v : u \in S, v \in S\}| \\
                                                            &= \internal(S).
\end{align*}

\noindent Hence, we get the following identity which, after preprocessing the indegrees and outdegrees using $\cO(n)$ queries, allows for the computation of $\internal(S)$ using $\cO(1)$ additional queries:

\[
    \internal(S) = \left( \sum_{u \in S} \outdegree(u) \right) - \query(S).
\] 

Finally, we define the subroutine $\cross$ which will be the most useful throughout. It simply denotes the number of edges between two sets of vertices:

\begin{align*}
    \cross(A, B) &:= |\{u \to v : u \in A, v \in B\}| + |\{u \to v : u \in B, v \in A\}| \\
                 &~= |\{u \to v : u \in A, v \in A \cup B\}| - |\{u \to v : u \in A, v \in A\}| + \\ 
                 &\qquad + |\{u \to v : u \in B, v \in A \cup B\}| - |\{u \to v : u \in B, v \in B\}| \\
                 &~= |\{u \to v : u \in A \cup B, v \in A \cup B\}| - |\{u \to v : u \in A, v \in A\}| - \\ 
                 &\qquad - |\{u \to v : u \in B, v \in B\}| \\
                 &~= \internal(A \cup B) - \internal(A) - \internal(B).
\end{align*}

\noindent We therefore get the identity

\begin{gather*}
    \cross(A, B) = \internal(A \cup B) - \internal(A) - \internal(B)
\end{gather*}

\noindent which means again that after computing all indegrees and outdegrees, $\cross$ can be computed with $\cO(1)$ additional cut queries.

While $\cross$ does not take edge directions into account, it can be used to find weak connectivity properties of the graph. Moreover, the edge direction limitation can be lifted if an algorithm chooses queries carefully. This lifting will be thoroughly studied in the context of topological sorting in \Cref{section:topsort}.

It might be useful to count the number of edges between a vertex $v$ and another set of vertices $S$. Therefore, as an abuse of notation, we may drop the braces around the singleton set consisting of $v$. That is, we have $\cross(v, S) := \cross(\{v\}, S)$.

We encapsulate the above identities into the following lemma.

\begin{lemma} \label{lemma:primitives}
    Let $G = (V, E)$ be a directed graph. After preprocessing $\indegree$ and $\outdegree$ using $\cO(n)$ cut queries, any calls to functions $\query$, $\querycomp$, $\internal$ and $\cross$ use $\cO(1)$ additional cut queries.
\end{lemma}

\begin{proof}

    \begin{gather}
        \querycomp(S) = \query(V \setminus S) \\
        \outdegree(v) = \query(\{v\}) \label{eq:indegree} \\
        \indegree(v) = \querycomp(\{v\}) \label{eq:outdegree} \\
        \internal(S) = \left( \sum_{u \in S} \outdegree(u) \right) - \query(S) \label{eq:internal} \\
        \cross(A, B) = \internal(A \cup B) - \internal(A) - \internal(B) \label{eq:cross}
    \end{gather}

    After precomputing $\indegree$ and $\outdegree$ for each vertex, every other formula uses $\cO(1)$ additional cut queries. Note that $\querycomp$ does not actually use the indegrees and outdegrees of other vertices.
\end{proof}

\paragraph{Discovering Edges} One natural question to ask is how to reconstruct the graph when it is sparse, as then edges would potentially be harder to find. We propose an extension to \cite[Lemma 2.1]{rubinstein_et_al:LIPIcs.ITCS.2018.39} and \cite[Proposition 2.10]{DBLP:conf/soda/JiangNS26} to return, given vertex sets $A$ and $B$, an edge with one endpoint in $A$ and one endpoint in $B$.

Since the function $\cross$ does not take the direction of edges into account, the lemma cannot determine the direction of the discovered edge. While the procedure described in the proof works for $G$, for simplicity, we state the lemma in terms of $\tilde{G}$.

\newcommand{\Discover}{\mathrm{Discover}}

\begin{lemma}[Discovering an Edge] \label{lemma:discover-edge}
    Let $A, B \subseteq V$ be two disjoint sets and $D \subseteq \tilde{E}$ such that $|D \cap (A \times B)| < \cross(A, B)$. There is a procedure that discovers an edge $uv \in (\tilde{E} \setminus D) \cap (A \times B)$ using $\cO(\log n)$ queries.
\end{lemma}

The set $D$ can be thought of as the set of edges that we have previously discovered. Then, \Cref{lemma:discover-edge} discovers a \textit{new} edge between $A$ and $B$.

The procedure described in the proof below consists of two binary searches, one for each end of an edge. At each point, we maintain as an invariant that there exists an edge between $A$ and $B$; then we shrink one of $A$ or $B$ in half until they both contain an element. Taking said element from both sets gives the sought-after edge.

\begin{proof}
    Let $\cross_D(A, B) = |((A \times B) \cap \tilde{E}) \setminus D|$. Notice that $\cross_D(A, B)$ is additive, that is, $\cross_D(A' \cup A'', B' \cup B'') = \cross_D(A', B) + \cross_D(A'', B) = \cross_D(A, B') + \cross_D(A, B'')$ where $A = A' \cup A''$, $B = B' \cup B''$ and $A', A'', B', B''$ are disjoint.

    We maintain as an invariant $\cross_D(A, B) > 0$. At each step, we choose either $A$ or $B$, split it in roughly equal parts and replace the chosen set with the half that maintains the invariant. Without loss of generality, we choose to split $A$ into disjoint $A', A''$ such that $|A'|$ and $|A''|$ differ by at most $1$. Since the sum $\cross_D(A', B) + \cross_D(A'', B) = \cross_D(A, B) > 0$, at least one of the terms must be positive. Hence, by replacing $A$ with $X \in \{A', A''\}$ such that $\cross_D(X, B) > 0$, the invariant is maintained.
    At the end of the procedure, let $A = \{u\}$ and $B = \{v\}$. Then, the procedure returns the edge $uv$.

    Since at each step, we have to compute $\cross_D(A', B)$ and $\cross_D(A'', B)$, and at each step, the size of either $A$ or $B$ gets halved, we must do $2\log |A| + 2\log |B| \leq 4\log n = \cO(\log n)$ queries.
    While not necessary, we can omit computing $\cross_D(A'', B)$, thus needing $2 \log n = \cO(\log n)$ queries.
\end{proof}

A consequence of \Cref{lemma:discover-edge} is that we can reconstruct $\tilde{G}$ in $\cO(m \log n)$ by iterating through each vertex and discovering all of its neighbors, thus reproducing a result of \cite[Section 2]{rubinstein_et_al:LIPIcs.ITCS.2018.39}.

\part{DAGs}
\section{Topological Sorting} \label{section:topsort}

We first focus on DAGs. This is a natural subclass of directed graphs, as its lack of cycles means that the directions of all edges can be determined using cut queries. 

Moreover, considering DAGs enables us to first compute a topological sort of a graph and use that as a central tool in deriving other algorithms \eg{for $s$-$t$ reachability or shortest paths}.
A topological sort gives queries more structure, as every edge would now go from left to right.
In particular, $\cross(A, B)$ is a function that counts the number of edges between $A$ and $B$ without taking their direction into account. 
On the other hand, if all the vertices of $A$ come before $B$ in the topological sort, we know that all the edges go from $A$ to $B$.
This type of selective query can then answer many reachability related questions. In particular, as we show in \Cref{section:applications}, we can easily solve the single-source reachability and single-source shortest path problems after computing a topological sort.

The remainder of this section will provide an efficient algorithm for computing the topological sort of a DAG and will prove its correctness.

\subsection{Topological Sort of DAGs}

In the classical setting, the usual framework used to compute a topological sort is that at each iteration, a vertex with zero indegree is appended to the result and then removed from the graph. To find such a vertex, an algorithm could simply iterate through every vertex in the remaining graph. As an optimization, after removing a vertex, we can check all of its out-neighbors again, since their indegree decreases by one.

\paragraph*{Adapting to the Cut Query Model.} Crucially, the two components that have to be adapted from the classical model are the computation of the indegree of a vertex and iterating through the out-neighbors of a vertex.

Assuming that we have so far computed a prefix of a topological sort $S'$, we define the \textit{remaining indegree} of a vertex $v$ to be its indegree in the graph $G - S'$. 

\begin{proposition}[Remaining Indegree of $v$] \label{prop:remainingindegree}
    Let $S'$ be the prefix of a topological sort and $v \in V \setminus S'$. The indegree of $v$ in $G - S'$ is $\indegree(v) - \cross(S', v)$.
\end{proposition}

\begin{proof}
    Since $S'$ is the prefix of a topological sort, there is no edge from $v$ to $S'$. Therefore, all the edges between $S'$ and $v$ are directed towards $v$.
\end{proof}

To iterate through the out-neighbors of a vertex $v$, we can apply \Cref{lemma:discover-edge} repeatedly. Since all the edges between $v$ and $V \setminus S'$ are directed towards the latter, all the discovered edges are between $v$ and its out-neighbors.

Combining the above, we obtain \Cref{alg:topsort-brute} which gives the adaptation of the topological sort algorithm into the cut query model.

\begin{algorithm}
    \begin{algorithmic}[1]
        \item[]\textbf{Input:} Cut query oracle function $\query$.
        \item[]\textbf{Output:} Topological ordering $S'$.

        \item[]

        \Procedure{Append}{$u$}
            \State Append $u$ to $S'$.
            \State Discover all outgoing edges $u \to v$ using \Cref{lemma:discover-edge} repeatedly.
            \For{each discovered edge $u \to v$}
                \If{$\indegree(v) - \cross(S', v) = 0$} 
                    \Call{Append}{$v$} 
                \EndIf
            \EndFor
        \EndProcedure

        \item[]

        \For{$u \in V$}
            \If{$\indegree(u) - \cross(S', u) = 0$}
                \Call{Append}{$u$}
            \EndIf
        \EndFor

        \State \textbf{return} $S'$

    \end{algorithmic}
    \caption{Simple algorithm for computing a topological sort in the cut query model.}
    \label{alg:topsort-brute}
\end{algorithm}

Since discovering an edge uses $\cO(\log n)$ queries by \Cref{lemma:discover-edge}, the total number of queries used by \Cref{alg:topsort-brute} is $\cO(m \log n)$.

\paragraph*{Reducing the Number of Queries.}

While \Cref{alg:topsort-brute} performs well on sparse graphs, its performance drops for graphs with $\Omega(n^2)$ edges \ie{dense graphs}, as it would need to reconstruct the entire graph. For the latter, the algorithm would use $\Theta(n^2 \log n)$ queries, which is worse than the quadratic baseline.

If a vertex $v$ has a small \textit{remaining indegree}, then, intuitively, we can afford to discover all of its incoming edges. The problematic case is when its \textit{remaining indegree} is large. To resolve this issue, we will check $v$ separately. Fortunately, if a vertex $v$ has a large \textit{remaining indegree} $\remindegree$, then $v$ must wait for the $\remindegree$ distinct in-neighbors to be placed first in the topological sort. This implies that a vertex with a large indegree can be checked only a small number of times.

\begin{observation} \label{obs:indegreehops}
    Let $G = (V, E)$ be a DAG and $v \in V$ a vertex. Then, for any topological sort $S_n = [s_1, \ldots, s_n]$ and integer $k$ such that $v \not\in S_k$, it holds that $v \not\in S_{k + \remindegree}$ where $\remindegree := \indegree(v) - \cross(S_k, v)$.
\end{observation}

\begin{proof}
    Since $S_n$ is a topological sort, by \Cref{prop:remainingindegree}, the value of $\remindegree$ is the indegree of $v$ in $G - S_k$. By assumption, $v$ cannot be part of the first $k$ elements of a topological sort.

    Let $u \in V \setminus S_k$ be a vertex such that there is an edge $u \to v$. Vertex 
$u$ must precede $v$ in the topological sort. Since $v$ must go after all $\remindegree$ such vertices, then $v$ must go after $k + \remindegree$ vertices in total.
\end{proof}

We classify vertices as \textit{heavy} or \textit{light} by their remaining indegrees in relation to an integer threshold $B$. If at some point in the algorithm, the remaining indegree of $v$ is \textit{observed} to be positive and less than or equal to $B$, we say that $v$ is \textit{light}. Otherwise, if the remaining indegree is greater than $B$, we say that $v$ is \textit{heavy}. If $\remindegree = 0$, we can append $v$ to the topological ordering $S'$. A vertex that is appended to the topological sort is neither \textit{heavy} nor \textit{light}. Note that a vertex is labeled as \textit{heavy/light} with regard to the \textit{observed} indegree. To illustrate this, consider a \textit{heavy} vertex that is currently being delayed as a consequence of \Cref{obs:indegreehops}. After removing multiple vertices, its indegree may become small. In this case, the vertex remains \textit{heavy} until it is checked again by the algorithm.

Intuitively, \textit{light} vertices are processed exactly as in \Cref{alg:topsort-brute}, while \textit{heavy} vertices are considered separately. When a \textit{heavy} vertex is checked and still has a large indegree, it is delayed by its remaining indegree so that it waits for its in-neighbors to be placed.

We combine the above insights into \Cref{alg:topsort}. At the beginning of the algorithm, we assume every vertex to be \textit{heavy}. The algorithm runs until a topological sort is computed. At each iteration, the algorithm checks every available heavy vertex and computes its remaining indegree according to \Cref{prop:remainingindegree}. If a vertex has large indegree $\remindegree$, it is ignored until at least $\remindegree$ more vertices are appended to $S'$, following \Cref{obs:indegreehops}. If it has a small positive indegree, the vertex is marked as light. If it has an indegree equal to zero, it is appended to the topological sort. Whenever a vertex is appended to the topological sort, all of its \textit{light} out-neighbors are discovered and then each one is checked to see whether it should also be appended to the topological sort, similarly to \Cref{alg:topsort-brute}.

\begin{algorithm}
    \begin{algorithmic}[1]
        \item[]\textbf{Input:} Cut query oracle function $\query$ and heaviness threshold parameter $B$.
        \item[]\textbf{Output:} Topological ordering $S'$.

        \item[]

        \State Compute $\indegree(v)$ and $\outdegree(v)$ to enable \Cref{lemma:primitives}.
        \State Mark all vertices as \textit{heavy}.

        \Procedure{Append}{u}
            \State Append $u$ to $S'$.
            \State Discover all the \textit{light} out-neighbors of $u$ using \Cref{lemma:discover-edge}.

            \For{each discovered edge $u \to v$}
                \State Let $\remindegree$ be the remaining indegree of $v$. 
                \If{$\remindegree = 0$} \Call{Append}{$v$} \EndIf
            \EndFor
        \EndProcedure

        \item[]

        \While{$|S'| < n$}
            \For{each \textit{heavy} vertex $u$}
                \State Let $\remindegree$ be the remaining indegree of $u$.
                \If{$\remindegree = 0$} \Call{Append}{$u$}
                \ElsIf{$\remindegree \le B$} Mark $u$ as \textit{light}.
                \Else{} Avoid checking $u$ until $\remindegree$ more vertices are added to $S'$.
                \EndIf
            \EndFor
        \EndWhile

    \end{algorithmic}
    \caption{Optimized algorithm for computing a topological sort in the cut query model.}
    \label{alg:topsort}
\end{algorithm}

For the remainder of this subsection, we prove that \Cref{alg:topsort} returns a topological sort and it does so efficiently (\Cref{thm:algtopsort}).

For the next lemma, we denote by $n'$ the size of $S'$ at any point in the algorithm. We define $S'_k := [s'_1, s'_2, \ldots, s'_k]$ to be the prefix of length $k$ of $S'$.

\begin{lemma} \label{lemma:topsortprefix}
    Let $G = (V, E)$ be a DAG and consider the sequence $S'_{n'} = [s'_1, \ldots, s'_{n'}]$ defined in \Cref{alg:topsort} at any point in time. Then, $S'_{n'}$ is the prefix of a topological ordering of $G$.
\end{lemma}

\begin{proof}
    Firstly, $S'_{n'}$ does not contain duplicates. That is because once a vertex is appended to $S'$, it is neither \textit{heavy} nor \textit{light} anymore. Therefore, it is never considered again to be appended.

    We will now prove the lemma by induction for each prefix of $S'$. That is, we prove that $S'_i$ is a prefix of a topological sort for every $i$ with $0 \le i \le n'$.

    If $i = 0$, since $G$ is a DAG, there exists a topological sort $S$ of $G$. The empty prefix $S'_0$ is trivially a prefix of $S$. 

    For the inductive step, we will prove that if the statement is true for $i - 1$, then it is also true for $i$.

    Consider the moment when $s'_i$ is added (Line~8 or Line~12). In both cases, since $S'_{i - 1}$ is a prefix of a topological sort, then, by \Cref{prop:remainingindegree}, the value $\remindegree$ is the indegree of $s'_i$ in graph $G - S'_{i - 1}$. Since $s'_i$ is appended only if $\remindegree = 0$, then there is no edge $u \to v$ such that $u \in V \setminus S'_{i}$.

    Consider the graph $G - S'_{i}$. Clearly, it must be a DAG, hence, it has a topological sort $S_{n - i}$. If we concatenate $S'_i$ and $S_{n - i}$, we obtain a valid topological sort, as every vertex appears exactly once and there cannot be an edge $u \to v$ that goes from right to left, as proven in the previous paragraph.
\end{proof}

\begin{lemma} \label{lemma:correctness}
    \Cref{alg:topsort} terminates and outputs a topological sort of DAG $G$. 
\end{lemma}

\begin{proof}
    It suffices to show only termination, as, by applying \Cref{lemma:topsortprefix}, correctness follows immediately.

    We show that at each iteration of the while-loop at Line~9, at least one element is added.

    Assume that at an iteration, the length of $S'$ is $n'$. Since $S'$ is the prefix of a topological sort, then there must be a vertex $v \in V \setminus S'_{n'}$ with remaining indegree zero.

    If $v$ is \textit{light}, consider the in-neighbor $u$ that was last added to the topological sort. Such a vertex exists because at Lines~11-14, $v$ becomes light if its remaining indegree is positive. Since $u$ is appended to $S'$ and then the for loop at Line~6 iterates through edge $u \to v$, vertex $v$ will then also be appended to $S'$.

    On the other hand, if $v$ is \textit{heavy}, it cannot hold that $v$ is ignored as a result of Line~14. Since $v$ has indegree zero in $G - S'$, then $S' + [v]$ is a valid prefix of a topological sort. Ignoring $v$ would then contradict \Cref{obs:indegreehops}. Therefore, $v$ must be iterated through by the for loop at Line 10, hence it will be appended to $S'$.
\end{proof}

Finally, we combine all the above lemmas to show that \Cref{alg:topsort} can be configured to run in a subquadratic number of queries. To do so, it remains to bound the number of queries used and then pick an appropriate parameter $B$. We show that the number of checks on \textit{heavy} vertices is inversely proportional to $B$ and that the number of checks on \textit{light} vertices is proportional to $B$. Then, as is customary for hybrid algorithms, the two bounds are balanced to obtain $B = \sqrt{n / \log n}$.

\algtopsort

\begin{proof}
    By \Cref{lemma:correctness}, \Cref{alg:topsort} terminates and returns a valid topological sort.

    \Cref{alg:topsort} queries the oracle at Line~1, Lines~5-7 and at Line~11. We will bound each type of query call separately in terms of $n$ and parameter $B$.

    By \Cref{lemma:primitives}, there are $\cO(n)$ queries at Line~1.
    The number of queries at Lines~5-7 is $\cO(\log n)$ for each discovered edge due to \Cref{lemma:discover-edge}. Since $\textsc{Append}(u)$ discovers only edges coming out of $u$ and $\textsc{Append}(u)$ is called exactly once for each vertex, any edge is discovered at most once.

    Consider the moment a vertex $v$ becomes \textit{light}. Vertex $v$ can only be marked light at Line~13. Let $k$ be the length of $S'$ at said moment. No edge $s'_i \to v$ is discovered, since Line~5 discovers only edges directed towards \textit{light} vertices. Since for any $i \le  k$, vertex $v$ is not \textit{light} when $\textsc{Append}(s'_i)$ is called, then we have that only edges with $u \to v$ with $u \not\in S'_k$ are discovered. On the other hand, this number of edges is the \textit{remaining indegree} defined in \Cref{prop:remainingindegree} and computed at Line~11. Since $v$ is marked as \textit{light}, it must hold that $\remindegree \le B$, hence, for each vertex, we discover at most $B$ incoming edges. By summing over all vertices, we get that the number of queries at Lines~5-7 is at most $\cO(nB \log n)$, the $\log$ factor coming from \Cref{lemma:discover-edge}.

    For the calls at Line~11, consider a vertex $u$ and every moment where $u$ is iterated through at Line~10. Let $k$ be the number of iterations through $u$ and let $n'_1, n'_2, \ldots, n'_k$ be the sizes of $S'$ at each of those moments. By \Cref{obs:indegreehops}, for any two adjacent values of $n'$, we have that $n'_{i + 1} - n'_{i} = \remindegree$ with $\remindegree$ being the remaining indegree of $u$ when $|S'| = n'_i$. By Line 14, it then holds that $\remindegree > B$, hence, $n'_{i + 1} - n'_{i} > B$. By summing up the inequalities for all adjacent values of $n'$, we get that $n'_{k} - n'_1 > (k - 1)B$. Since $n'_{k} \le n$ and $n'_1 = 0$, by rearranging terms we obtain $k < n / B + 1$. Therefore, for each vertex $u$, Line~11 is executed $\cO(n / B)$ times. Consequently, the total number of queries from Line~11 is $\cO(n^2 / B)$.

    By adding the two bounds together, we obtain $\cO(n) + \cO\left( n^2 / B \right) + \cO\left( nB \log n \right)$ queries. If we set $B$ to $\sqrt{n / \log n}$, we obtain an algorithm that performs $\cO\left( n \sqrt{n \log n}  \right) $ queries.
\end{proof}

Notably, \Cref{thm:algtopsort} shows that computing a topological sort can be done in a subquadratic number of queries, beating the trivial baseline algorithm. As a consequence, other problems on DAGs can also be solved subquadratically by exploiting the topological sort. Such algorithms are presented in \Cref{section:applications}.

The ideas in this section can be extended to topologically sort the SCCs of a general directed graph once the SCCs are known.
This implies that, once the strongly connected components are known, we can compute $s$-$t$-reachability in polynomially subquadratic queries.
Given our current knowledge, this is subsumed by the following sections (concretely, \Cref{thm:det}).
We still include it in \Cref{sec:toposort-scc} as it provides a different avenue of attack for improved reachability algorithms for general directed graphs.

\section{Deriving Algorithms for DAGs from Topological Sorting} \label{section:applications}

One common technique on DAGs is to compute the answer to a problem via dynamic programming. By defining a subproblem for each vertex, we can then process the vertices in topological order and apply a recurrence relation. Two examples considered in this section are for \textit{single-source reachability} and \textit{single-source shortest paths}.

\subsection{Single-Source Reachability}

\Cref{alg:reachability} solves the single-source reachability problem. It considers every vertex in topological order and decides for each whether it is reachable from the source vertex $s$ using $\cO(1)$ queries. Given that all edges are directed from left to right in a topological ordering, it suffices to find whether a vertex $v$ is connected to an already known reachable vertex.
To this end, the set $R$ maintains the reachable vertices observed so far.

\correachability

\begin{algorithm}[hhbp]
    \begin{algorithmic}[1]

    \item[]\textbf{Input:} Cut query oracle function $\query$ and source vertex $s$.
    \item[]\textbf{Output:} Set $R \subseteq V$ of vertices reachable from $s$.

    \State Let $v_1, v_2, \ldots, v_n$ be a topological sort computed using \Cref{thm:algtopsort}.
    \State $R \gets \emptyset$ 

    \For{$i \in [1, n]$}
        \If{$v_i = s$ \textbf{or} $\cross(R, v_i) > 0$}
            \State $R \gets R \cup \{v_i\}$
        \EndIf
    \EndFor

    \State{\textbf{return} $R$}

    \end{algorithmic}
    \caption{Computing reachability in the cut query model.}
    \label{alg:reachability}
\end{algorithm}

\begin{proof}
    Consider \Cref{alg:reachability}. Clearly, it terminates, as \Cref{alg:topsort} also terminates, and Lines 2-9 use $\cO(n)$ queries.

    We will prove by induction that for all $i$ with $0 \le i \le n$, after $i$ iterations, a vertex $x$ among $v_1, \ldots, v_i$ is added to $R$ if and only if $x$ is reachable from $s$.

    For the base case when $i = 0$, the statement trivially holds.

    We now prove the inductive step. By induction, all vertices $v_{j < i}$ are added to $R$ if and only if they are reachable from $s$. It therefore suffices to show the equivalence only for $v_i$.

    If $v_i$ is added to $R$, it then holds that either $v_i = s$, in which case the equivalence trivially holds, or $\cross(R \setminus \{v_i\}, v_i) > 0$ (Line 4). Since the vertices are considered according to a topological ordering, every edge is directed from left to right, hence the inequality implies that there is an edge $u \to v_i$ with $u \in R \setminus \{v_i\}$. By induction, there must also be a path $s \rightsquigarrow u$. This means that there is also a path from $s$ to $v_i$.

    Conversely, the statement trivially holds for $v_i = s$. Otherwise, assume there is a path from $s$ to $v_i$ and let $u$ be the last vertex on said path before $v_i$. Since the vertices are in topological order and there is an edge $u \to v_i$, then $u$ must come before $v_i$ in the topological ordering, hence, $u \in \{v_1, \ldots, v_{i - 1}\}$. Since $u$ is also reachable from $s$, then, by the induction hypothesis, we have that $u \in R \setminus \{v_i\}$. The existence of the edge $u \to v_i$ implies that $\cross(R \setminus \{v_i\}, v_i) > 0$. This concludes the proof of the inductive step.

    Therefore, it holds that the set $R$ contains all vertices reachable from $s$.

    In total, \Cref{alg:reachability} uses $\cO(n \sqrt{n \log n}) + \cO(n) = \cO(n \sqrt{n \log n} )$ cut query oracle calls, the latter term coming from the calls at Lines 3-5.
\end{proof}

Note that the bottleneck of \Cref{alg:reachability} is finding the topological ordering due to \Cref{thm:algtopsort}, as the pass to actually compute reachability uses in total a linear number of queries.

\subsection{Single-Source Shortest Paths}
\label{sec:dag-shortest-paths}
In this subsection, we will focus on \textit{single-source shortest paths} in the cut query model for DAGs. Given a source vertex $s$, we aim to compute the shortest path between $s$ and every other vertex. To this end, for each vertex $v$ we compute the following information:

\begin{description}[noitemsep]
  \item[$D\left(v\right)$] distance from $s$ to $v$, or $\infty$ if there is no path from $s$ to $v$, and
  \item[$\pred\left(v\right)$] last vertex before $v$ in the shortest path from $s$ to $v$, or $\bot$ if there is no such path.
\end{description}

We store these values in the vectors $\Dvec=(D(v))_{v\in V}$ and $\predvec=(\pred(v))_{v\in V}$.

Then, to reconstruct the path from $s$ to $t$, we follow $\pred$ until we reach $s$. If we interpret the $\pred$ function as a tree rooted at $s$, then the shortest path from $s$ to $t$ in $G$ is given by the path from $s$ to $t$ in this tree.
In the classical model, the recurrence formula is:
\begin{gather*}
    D(s) = 0, \pred(s) = \bot, \\
    D(v) = \begin{cases}
        \infty & \text{if there is no $s$-$v$ path}, \\
        \min_{u \to v} D(u) + 1 & \text{otherwise},
    \end{cases} \\
    \pred(v) = \begin{cases}
        \bot & \text{if there is no $s$-$v$ path}, \\
        \arg \min_{u \to v} D(u) & \text{otherwise}.
    \end{cases}
\end{gather*}
To compute the recurrence in the classical model, we can iterate through each vertex in the topological order, then for each vertex we iterate through its in-neighbors.

Deciding whether there is a path from $s$ is easy, as \Cref{thm:reachability} can be used to filter out all unreachable vertices. The difficulty in the cut query model is that iterating through every in-neighbor implies reconstructing the graph.
To resolve this issue, we can apply a modified version of \Cref{lemma:discover-edge} to discover the edge that yields the shortest path. 

\begin{lemma}[Discovering the Minimal Edge] \label{lemma:discover-minimal-edge}
    Let $A \subseteq V$ be a set, $v \in V$ be a vertex and $f : A \to \mathbb{R}$ a function. There is a procedure that discovers the edge $uv$ with $u \in A$ and $f(u)$ minimal or reports that there is no such edge using $\cO(\log n)$ queries.
\end{lemma}

\begin{proof}
    If $\cross(A, v) = 0$, then there is no edge between $A$ and $v$; hence the procedure can report this.
    On the other hand, similarly to the proof of \Cref{lemma:discover-edge}, we maintain as an invariant that $\cross(A, v) > 0$ and split $A$ in half.

    In addition to \Cref{lemma:discover-edge}, to find the minimal edge, we order the vertices of $A$ in order of $f$. Then, we split $A$ so that $A'$ contains the vertices with smaller $f$ values. If $\cross(A', v) > 0$, we replace $A$ with $A'$. Otherwise, we replace $A$ with $A''$. Doing so ensures that the procedure always selects the set that contains the minimal edge.
    Since at each step, the size of $A$ halves, the procedure uses $\cO(\log n)$ queries.
\end{proof}

\begin{algorithm}[H]
    \begin{algorithmic}[1]

    \item[]\textbf{Input:} Cut query oracle function $\query$ and source vertex $s$.
    \item[]\textbf{Output:} Values $D(v)$ and $\pred(v)$ for every vertex.

    \State Let $v_1, v_2, \ldots, v_n$ be a topological sort computed using \Cref{thm:algtopsort}.
    \State Remove vertices from the ordering that are not reachable from $s$ using \Cref{thm:reachability}.
    \State $D(v_i) \gets \infty, \pred(v_i) \gets \bot$ for all unreachable vertices.
    \State $D(s) \gets 0, \pred(s) \gets \bot$

    \item[]

    \For{each reachable vertex $v_i$ excluding $s$}
        \State Find the edge $u \to v_{i}$ with $D(u)$ minimal and $u \in \{v_1, v_2, \ldots, v_{i - 1}\}$ using \Cref{lemma:discover-minimal-edge}.
    \State $D(v_i) \gets D(u) + 1$
        \State $\pred(v_i) \gets u$
    \EndFor

    \State \textbf{return} $\Dvec, \predvec$

    \end{algorithmic}
    \caption{Algorithm that computes all shortest paths from $s$.}
    \label{alg:sssp}
\end{algorithm}

\algsssp

\begin{proof}
    We will prove correctness of \Cref{alg:sssp} by induction: for all $i$ with $0 \le i \le n$, the entries $D(v_j)$ and $\pred(v_j)$ are correctly computed for all $1\le j\le i$. Due to \Cref{alg:reachability}, we assume without loss of generality that all the vertices of $G$ are reachable from $s$.

    For the base case when $i = 0$, the statement trivially holds.

    For the inductive step, we prove that if the statement holds for $i - 1$, then it also holds for $i$. It suffices to show that $D(v_i)$ and $\pred(v_i)$ are correctly computed.

    The statement holds trivially for $v_i = s$. 
    If $v_i \neq s$ holds, assume that the length of the shortest path between $s$ and $v_i$ is $l$. Let $s = p_1, p_2, \ldots, p_{l + 1} = v_i$ be the shortest path from $s$ to $v_i$.
    Let $u = p_{l}$. The path $p_1, \ldots, p_{l}$ must also be a shortest path to $u$. By induction, we then have that $D(u) = l - 1$. Since Line~6 selects the vertex with an incoming edge with minimal $D(u)$, we then have that:
    \[
        D(v_i) = \min_{u \to v_i} D(u) + 1 \le  D(u) + 1 = l.
    \] 
    For the converse statement, let $u = \pred(v_i)$. We append the edge $u \to v_i$ to the shortest path from $s$ to $u$ to obtain another path of length $D(v_i) = D(u) + 1$. By definition, $D(v_i)$ is at least the length of the shortest path from $s$ to $v_i$.
    Therefore, $D(v_i)$ is computed correctly. Since there is an edge $\pred(v_i) \to v_i$ and $D(v_i) = D(\pred(v_i)) + 1$, we have that $\pred(v_i)$ is also computed correctly. This concludes the proof of correctness.

    \Cref{alg:sssp} uses $\cO(n \sqrt{n \log n}) + \cO(n) + \cO(n \log n) = \cO(n \sqrt{n \log n})$ queries. The first two terms come from \Cref{thm:algtopsort} and \Cref{thm:reachability} respectively. The last term comes from the pass to compute the distance to each vertex. For each vertex, we use \Cref{lemma:discover-minimal-edge} to find the edge that gives the best distance.
\end{proof}

Note that \Cref{alg:sssp} can be adapted to compute different recurrences. For instance, to compute the longest path, we take the maximum out of all neighbors. To find the length of the shortest path towards a sink $s$, we can process vertices in reverse topological order. Generally, any recurrence that takes as an aggregate the minimum or the maximum element out of all in- or out-neighbors can be computed using $\cO(\log n)$ queries for any vertex.

\part{General Graphs}
\section{The Plan for General Directed Graphs and Pseudo-Orientations}\label{sec:orientations}
Our efforts to derive algorithms for general directed graphs crucially revolve around the notion of a pseudo-orientation.
We first motivate and define this concept, define it.

Call the multigraph obtained by forgetting the direction of every edge in $G$ its \emph{skeleton}, and let $M_{uv}\in\{0,1,2\}$ be the number of edges joining distinct vertices $u,v$. In particular, the two arcs $u\to v$ and $v\to u$ give two parallel skeleton edges.
Imagine we could learn an orientation of the skeleton of $G$ with the same out-degrees as $G$.
Then, it would be easy to deduce $s$-$t$ reachability (and also things like all-pairs reachability, $s$-$t$ minimum cut or global minimum cut).
Unfortunately this is impossible as we cannot hope to learn the entire skeleton in polynomially subquadratic queries (as \cite{grebinskiOptimalReconstructionGraphs2000a} first showed).

\begin{remark}\label{rem:skeleton-information}
  Consider a balanced bipartite DAG, that is a graph with two vertex sets $V_1, V_2$ each of $n/2$ vertices where every edge is from $V_1 \times V_2$.
  A simple information theoretic argument shows that if each edge is present independently with probability $\nicefrac{1}{2}$, reconstructing the graph requires $\Omega(n^2 / \log n)$ cut queries.
  The skeleton (and thus also an orientation) would reveal the entire graph and thus also requires $\Omega(n^2 / \log n)$ cut queries.

  For the same reason we cannot hope for an all-pairs reachability algorithm with polynomially subquadratic queries.
\end{remark}

Thus, we settle for something weaker.
A \emph{pseudo-orientation} is a real-valued vector $\zvec$ over $V^2$ such that for each $u,v \in V$, we have $z_{uv} + z_{vu} = 1$.
Intuitively, $z_{uv}$ represents the fraction of the $M_{uv}$ edges between $u$ and $v$ that point in the direction $u \to v$.
We say $\zvec$ has \emph{integral degrees} if the orientation defined by $\zvec M$ is integral.
(Equivalently, each $z_{uv} \in \{0,\nicefrac{1}{2},1\}$ and takes the value $\nicefrac{1}{2}$ only if $M_{uv} = 2$.)

\begin{figure}[tbp]
\centering
\begin{tikzpicture}[node distance=2cm, on grid, auto,
    vu arc/.append style={draw=VUBlue}]

    \foreach \prefix/\shift in {A/0,B/5.5}{
        \node[vu vertex] (\prefix a) at (\shift-2.3,0) {$a$};
        \node[vu vertex] (\prefix b) at (\shift-1,1.5) {$b$};
        \node[vu vertex] (\prefix c) at (\shift-1,-1.5) {$c$};
        \node[vu vertex] (\prefix d) at (\shift+0.8,1) {$d$};
        \node[vu vertex] (\prefix e) at (\shift+0.8,-1) {$e$};
    }

    \draw[vu arc] (Aa) to[bend left] (Ab); 
    \draw[vu arc] (Ab) to[bend left] (Aa); 
    \draw[vu arc] (Ab)--(Ac); 
    \draw[vu arc] (Ac) -- (Aa);
    \draw[vu arc] (Ad) to[bend left] (Ae); 
    \draw[vu arc] (Ae) to[bend left] (Ad);
    \draw[vu arc] (Ab)--(Ad); 
    \draw[vu arc](Ac)--(Ae);

    \draw[vu arc, color=gray,opacity=0.5] (Ba)--(Bd);
    \draw[vu arc, color=gray,opacity=0.5] (Be)--(Ba);
    \draw[vu arc, color=gray,opacity=0.5] (Be)--(Bb);
    \draw[vu arc, color=gray,opacity=0.5] (Bc)--(Bd);
    \draw[vu arc, color=VUOrange] (Ba) to[bend left, color=VUOrange] (Bb);
    \draw[vu arc, color=VUOrange] (Ba) to[bend right, color=VUOrange] (Bb);
    \draw[vu arc] (Bb)--(Bc); 
    \draw[vu arc, color=VUOrange] (Ba)--(Bc);
    \draw[vu arc] (Bd) to[bend left] (Be); 
    \draw[vu arc] (Be) to[bend left] (Bd);
    \draw[vu arc] (Bb)--(Bd); 
    \draw[vu arc, color=VUOrange](Be)--(Bc);

    \node[vu vertex] (v) at (11, 0) {$v$};
    \foreach \index in {1, 2, 3, 4, 5}{
        \node[vu vertex] (u\index) at (11-1.5,0+3-\index) {$u_{\index}$};
        \node[vu vertex] (w\index) at (11+1.5,0+3-\index) {$w_{\index}$};
    }

    \draw[vu arc, color=gray, opacity=0.5] (u1)--(v);
    \draw[vu arc] (u2)--(v);
    \draw[vu arc, color=gray, opacity=0.5] (u3)--(v);
    \draw[vu arc, color=gray, opacity=0.5] (u4)--(v);
    \draw[vu arc] (u5)--(v);

    \draw[vu arc] (v)--(w1);
    \draw[vu arc, color=gray, opacity=0.5] (v)--(w2);
    \draw[vu arc] (v)--(w3);
    \draw[vu arc, color=gray, opacity=0.5] (v)--(w4);
    \draw[vu arc] (v)--(w5);

\end{tikzpicture}
\caption{Consider the directed graph on the left under the cut query model. We have $M_{ab} = M_{ba} = 2$, $M_{ac} = M_{ca} = 1$ and $M_{be} = M_{eb} = 0$.
The figure in the middle represents the graph skeleton with edges oriented according to a pseudo-orientation $\zvec$. A pseudo-orientation can be seen as an "embedding" of the original graph into a tournament graph. The edges in gray are edges that are not present in the original graph. Edges highlighted in orange are edges that are present in the graph, but wrongly oriented. For this pseudo-orientation, we have $z_{ac} = 1, z_{ca} = 0$, $z_{ab} = 1, z_{ba} = 0$, $z_{de} = z_{ed} = \nicefrac 12$ and $z_{eb} = 1, z_{be} = 0$. In the last figure, under a given pseudo-orientation, $v$ knows that its in-neighbors are the vertices with $z_{uv} = 1$, but does not know which of the edges $u_i \to  v$ actually exist. The same goes for out-neighbors.}
\label{fig:pseudoorientation}
\end{figure}

Our approach for general directed graphs is to construct a known pseudo-orientation whose induced skeleton orientation has the same out-degrees as $G$.
We first choose a pseudo-orientation with integral degrees without learning the skeleton that has a bounded error when compared to $\dvec$.
We then repair the degree errors by path reversals using a modified version of Dinic's algorithm.
Finally, in \Cref{sec:consequences}, we show how to use this pseudo-orientation with the same out-degrees as $G$ to query $s$-$t$ reachability with $\cO(n\log n)$ extra queries.
Intuitively, these extra queries supply \emph{some} of the skeleton multiplicities, which we show are enough for a given $s$ and $t$.

\paragraph{Outline.}
In \Cref{sec:det}, we construct a fractional pseudo-orientation by greedy averaging and round its entries in small vertex blocks. The output is a pseudo-orientation $\yvec$ with integral degrees whose out-degree vector has $\ell_1$ error $\cO(n^{8/5})$. The algorithm also computes its degree-error vector using aggregate cut queries; it does not need to learn each multiplicity.

In \Cref{sec:repair}, we apply path reversals directly to the rounded pseudo-orientation, querying multiplicities on the paths. For initial $\ell_1$ error $\cO(n^{8/5})$, repair uses $\cO(n^{9/5}\log n)$ queries and returns a pseudo-orientation with integral degrees whose out-degree vector is $\dvec$.

The two stages balance to use $\cO(n^{9/5}\log n)$ queries. \Cref{sec:consequences} shows how to use this pseudo-orientation to compute $s$-$t$ reachability and gives further applications.

\section{A Pseudo-Orientation with Small Degree Error}\label{sec:det}
We first construct a pseudo-orientation whose induced out-degree vector has $\ell_1$ degree error $\cO(n^{8/5})$.
We do this by computing a fractional orientation and then rounding it to have integral outdegrees.
Both of these steps follow a simple greedy approach.

\subsection{A Fractional Pseudo-Orientation via Greedy Averaging}
The idea is to average several pseudo-orientations whose induced degree errors compensate for one another. In each round, we give outgoing degree preferentially to vertices that have received too little in earlier rounds, relative to their targets. An individual choice may have large degree errors; what matters is the error of the average. Figure~\ref{fig:averaging} gives a two-round example.

In more detail, the algorithm works in rounds $t = 1, \dots, T$.
We will maintain an error vector $\hvec^t$ over the nodes that captures which nodes have received too much or too little out-degree so far.
Initialize $\hvec^0 = \zerovec$.
Now, in each round:
\begin{enumerate}
  \item Choose an integral pseudo-orientation $\zvec^t$ by directing each unordered pair from the endpoint with the smaller entry of $\hvec^{t-1}$ to the endpoint with the larger entry, breaking ties arbitrarily.
  \item\label{step:greedy:out} Compute the out-degrees $\outvec^t_{\zvec}$ of the induced orientation as, for each vertex $u$,
        \[
          \out^t_{\zvec}(u)=\cross\bigl(u,\{v\ne u:z^t_{uv}=1\}\bigr).
        \]
  \item Compute $\hvec^t$ as the cumulative out-degree errors, that is, as $\hvec^t \coloneqq \sum_{i=1}^t (\outvec^i_{\zvec} - \dvec)$.
\end{enumerate}

After $T$ rounds, we define our final fractional pseudo-orientation as the average of the pseudo-orientations picked in all rounds, that is, as $\xvec \coloneqq \left(\sum_{t=1}^T \zvec^t\right) / T$. See \Cref{fig:averaging} for an illustration of this process.

\begin{lemma}
The above greedy algorithm uses $\cO(nT)$ cut queries.
\end{lemma}

\begin{proof}
  As usual, we assume that we compute all out-degrees using $\cO(n)$ cut queries once.
  The only other step where we query is Step~\ref{step:greedy:out}, where we perform $\cO(n)$ $\cross$ queries per round.
  Recall from \Cref{section:subroutines} that each $\cross$ query costs $\cO(1)$ cut queries.
\end{proof}

\begin{lemma}\label{lem:averaging-error}
    The degree-error vector of $\xvec$ is $\hvec^T/T$. That is, $\outvec_{\xvec} - \dvec = \hvec^T / T$.

\end{lemma}

\begin{proof}
By definition, $x_{uv}=T^{-1}\sum_{t=1}^T z^t_{uv}$. Thus, for every vertex $u$,
\begin{align*}
 \out_{\xvec}(u)-\outdegree(u)
 &=\sum_{v\ne u}M_{uv}\left(\frac1T\sum_{t=1}^T z^t_{uv}\right)-\outdegree(u)\\
 &=\frac1T\sum_{t=1}^T\bigl(\out^t_{\zvec}(u)-\outdegree(u)\bigr)
 =\frac{h^T_u}T.
\end{align*}
Hence $\outvec_{\xvec}-\dvec=\hvec^T/T$.
\end{proof}

With this tool in hand, we can finally bound the total error of $\xvec$. We start by bounding the squared $\ell_2$ error.

\begin{figure}[tbp]
\centering
\begin{tikzpicture}[x=1cm,y=1cm,vu fractional/.append style={draw=VUBlue}]
  \foreach \prefix/\shift in {A/0,B/4.7,C/9.4}{
    \node[vu vertex] (\prefix u) at (\shift+1.3,1.5) {$u$};
    \node[vu vertex] (\prefix v) at (\shift,0) {$v$};
    \node[vu vertex] (\prefix w) at (\shift+2.6,0) {$w$};
  }
  \node[vu title] at (1.3,2.3) {Round 1: $\zvec^1$};
  \node[vu title] at (6,2.3) {Round 2: $\zvec^2$};
  \node[vu title] at (10.7,2.3) {Average: $\xvec$};
  \draw[vu arc] (Au)--(Av); \draw[vu arc] (Au)--(Aw); \draw[vu arc] (Av)--(Aw);
  \draw[vu arc] (Bv)--(Bu); \draw[vu arc] (Bw)--(Bu); \draw[vu arc] (Bw)--(Bv);
  \foreach \from/\to in {Cu/Cv,Cu/Cw,Cv/Cw}{
    \draw[vu fractional,-{Stealth[length=1.8mm]}] (\from) to[bend left=12] (\to);
    \draw[vu fractional,-{Stealth[length=1.8mm]}] (\to) to[bend left=12] (\from);
  }
  \node[vu text] at (1.3,-.7) {$\outvec^1_{\zvec}=(2,1,0)$\\$\hvec^1=(1,0,-1)$};
  \node[vu text] at (6,-.7) {$\outvec^2_{\zvec}=(0,1,2)$\\$\hvec^2=\zerovec$};
  \node[vu text] at (10.7,-.7) {$\outvec_{\xvec}=(1,1,1)$\\$\nicefrac{1}{2}$ in each direction};
  \draw[vu arc,black!55] (3.15,.65)--(3.95,.65);
  \draw[vu arc,black!55] (7.85,.65)--(8.65,.65);
\end{tikzpicture}
\caption{Two rounds of greedy averaging on a triangle with target degrees $(1,1,1)$ in vertex order $(u,v,w)$. Their average attains the target degrees; dashed arrows represent half of each skeleton copy in each direction.}
\label{fig:averaging}
\end{figure}

\begin{lemma}[Bounding the Error]\label{lem:weakavg}
We have $\norm{\outvec_{\xvec}-\dvec}_2^2\le4n^3/T$.
\end{lemma}
\begin{proof}
  Write $\evec^t=\outvec^t_{\zvec}-\dvec$.
  We first bound the squared $\ell_2$ norm of $\hvec^t$. By opening up the square, we get 
  \[
      \norm{\hvec^{t}}_2^2 = \norm{\hvec^{t - 1} + \evec^{t}}_2^2 = \norm{\hvec^{t - 1}}_2^2 + \norm{\evec^{t}}_2^2 + 2 \langle \hvec^{t - 1}, \evec^{t} \rangle.
  \] 

  On the other hand, we can show that $\langle \hvec^{t - 1}, \evec^{t} \rangle \le 0$. Consider the original graph $G$. Trivially, its error is zero, hence the inner product is also zero. Starting from $G$, flip edges $u \to v$ that do not agree with $\zvec^t$. This reversal decreases the out-degree of $u$ by one and increases that of $v$ by one. This increases the inner product by $-h^{t - 1}_u + h^{t - 1}_v$. Since $\zvec^t$ always points edge  $uv$ towards the node with larger error, we have that
  \[
      -h^{t - 1}_u + h^{t - 1}_v \le 0.
  \] 
  By reversing all the edges that disagree with $\zvec^t$, the error of $G$ transforms into that of the pseudo-orientation induced by $\zvec^t$. Summing over all reversals gives $\langle \hvec^{t - 1}, \evec^{t} \rangle \le 0$, hence, $\norm{\hvec^t}_2^2 \le \norm{\hvec^{t - 1}}_2^2 + \norm{\evec^t}_2^2$.

  Since $|e^t_u| \le 2n$ for any node $u$, we get $\norm{\evec^t}_2^2 \le 4n^3$. By applying the last two inequalities repeatedly starting with $\hvec^0 = \zerovec$, we get that $\norm{\hvec^T}_2^2 \le 4n^3T$. Dividing by $T^2$ and applying \Cref{lem:averaging-error} gives
  \[
      \norm{\outvec_{\xvec} - \dvec}_2^2 = \norm{\hvec^T / T}_2^2 \le \frac{4n^3}{T}. \qedhere
  \] 

\end{proof}

\subsection{Greedy Rounding in Blocks}
We now round $\xvec$ to a pseudo-orientation $\yvec$ with integral degrees. This uses a greedy direction choice for each pair and aggregate queries to update the degree errors.

Partition $V$ into $k=\lceil n/b\rceil$ blocks $V_1,\ldots,V_k$, each containing at most $b$ vertices. We fix $b$ later. Process every block pair $(V_i,V_j)$ with $1\le i\le j\le k$, including each block paired with itself. Each unordered vertex pair is processed exactly once.

Initialize $\yvec=\xvec$ and maintain its degree-error vector $\evec=\outvec_{\yvec}-\dvec$. Observe that initially $\evec = \hvec^T/T$. Unprocessed entries retain their values from $\xvec$; processed entries are rounded to zero or one. Keep $\evec$ frozen while choosing directions within a block pair; afterward, recompute the errors of the vertices in $V_i\cup V_j$.

For each unordered pair $\{u,v\}$ processed in the current block pair, choose the direction from the vertex with smaller error in $\evec$ to the larger (breaking ties arbitrarily). For example, if $u \to v$ is chosen, set $y_{uv}=1$ and $y_{vu}=0$.
Figure~\ref{fig:blocks} illustrates one such step, including a pair of multiplicity two.

After all block pairs have been processed, $\yvec$ has entries in $\{0,1\}$, so it has integral degrees. The algorithm outputs $\yvec$ and its exact degree-error vector $\evec$. Both copies of a pair with $M_{uv}=2$ receive the chosen direction.

First, we bound how many queries this rounding requires to update its error vector.

\begin{lemma}[Query Complexity of Block Rounding]\label{lem:block-queries}
For the average $\xvec$ above and any integer $1\le b\le n$, greedy block rounding maintains its exact degree-error vector by recomputing the out-degrees of at most $2b$ vertices after each block pair, using $\cO((n^2/b)\log(T+1))$ cut queries in total.
\end{lemma}
\begin{proof}
The initial degree-error vector $\hvec^T/T$ is already known. Rounding a block pair $(V_i,V_j)$ changes only the out-degrees of vertices in $V_i\cup V_j$. For each such vertex $u$, we recompute its weighted out-degree $\out_{\yvec}(u)=\sum_{v\ne u}M_{uv}y_{uv}$ and set $e_u=\out_{\yvec}(u)-\outdegree(u)$.

Every current entry $y_{uv}$ is either its initial value $x_{uv}$ or a rounded bit. Thus all weights remain multiples of $1/T$, and $k_{uv}:=Ty_{uv}\in\{0,\ldots,T\}$ is a known integer. We simulate the weighted $\cross$ query for the out-degree using one ordinary $\cross$ query per bit:
\[
 \out_{\yvec}(u)
 =\frac1T\sum_{\ell=0}^{\lfloor\log_2 T\rfloor}2^\ell\,
 \cross\bigl(u,\{v\ne u:\text{bit }\ell\text{ of }k_{uv}\text{ is }1\}\bigr).
\]
All sets on the right are known from $\yvec$, and each $\cross$ query uses $\cO(1)$ cut queries by \Cref{lemma:primitives}. Hence recomputing one out-degree takes $\cO(\log(T+1))$ queries. Since $\outdegree(u)$ is already known, this also gives the exact error.

There are at most $2b$ affected vertices per block pair and $\cO((n/b)^2)$ block pairs, giving the bound.
\end{proof}

\begin{figure}[tbp]
\centering
\begin{tikzpicture}[x=1cm,y=1cm]
  \foreach \prefix/\shift in {L/0,R/7.6}{
    \draw[vu panel,draw=VUBlue,fill=VUBlue2!12] (\shift-.85,-1.15) rectangle (\shift+.85,2.15);
    \draw[vu panel,draw=VUOrange,fill=VUOrange2!10] (\shift+2.95,-1.15) rectangle (\shift+4.65,2.15);
    \node[vu title] at (\shift,2.5) {$V_i$};
    \node[vu title] at (\shift+3.8,2.5) {$V_j$};
    \node[vu vertex] (\prefix u) at (\shift,1.5) {$u$};
    \node[vu vertex] (\prefix p) at (\shift,0) {$p$};
    \node[vu vertex] (\prefix v) at (\shift+3.8,1.5) {$v$};
    \node[vu vertex] (\prefix q) at (\shift+3.8,0) {$q$};
  }
  \node[vu title] at (1.9,3.2) {Before: Fractions};
  \node[vu title] at (9.5,3.2) {After: Direction Bits};
  \draw[vu pair] (Lu) to[bend left=12] (Lv);
  \draw[vu pair] (Lu) to[bend right=12] (Lv);
  \draw[vu pair] (Lu)--(Lq); \draw[vu pair] (Lp)--(Lv);
  \draw[vu arc] (Ru) to[bend left=12] (Rv);
  \draw[vu arc] (Ru) to[bend right=12] (Rv);
  \draw[vu arc] (Ru)--(Rq); \draw[vu arc] (Rp)--(Rv);
  \node[font=\small,below=1pt of Lu] {$e_u=-\tfrac12$};
  \node[font=\small,below=1pt of Lp] {$e_p=-\tfrac12$};
  \node[font=\small,below=1pt of Lv] {$e_v=\tfrac12$};
  \node[font=\small,below=1pt of Lq] {$e_q=\tfrac12$};
  \node[vu text,anchor=north,text width=6.1cm,inner sep=0pt] at (1.9,-1.45) {Each copy is split equally\\between the two directions.\\[3pt]$d_{ij}(u)=2+1=3$};
  \node[vu text,anchor=north,text width=6.1cm,inner sep=0pt] at (9.5,-1.45) {Each copy points toward\\the larger frozen error.\\[3pt]$\Delta e_u=\tfrac32\le d_{ij}(u)$};
  \draw[vu arc,black!55] (5,1)--node[above,font=\small] {round} (6.4,1);
\end{tikzpicture}
\caption{Rounding a block pair with $b=2$ and $m_{ij}=4$. Each skeleton copy is directed toward the endpoint with larger frozen error, giving error vector $(1,0,-1,0)$ in vertex order $(u,p,v,q)$.}
\label{fig:blocks}
\end{figure}

Next, let's bound the error of $\yvec$.

\begin{lemma}\label{lem:block}
The algorithm constructs a pseudo-orientation $\yvec$ with integral degrees satisfying
\[
 \norm{\outvec_{\yvec}-\dvec}_2^2\le\norm{\outvec_{\xvec}-\dvec}_2^2+4bm.
\]
\end{lemma}
\begin{proof}
    For one block pair $(V_i,V_j)$, let $\Delta\evec$ be the change in the degree-error vector caused by rounding this block pair, so the update is $\evec\leftarrow\evec+\Delta\evec$. It is easy to see that this change in error is the same as the change in outdegrees (i.e. $\outvec_{\yvec} = \outvec_{\xvec} + \Delta\evec$).

    Similarly to \Cref{lem:weakavg}, we prove that
    \[
        \norm{\evec + \Delta\evec}_2^2 = \norm{\evec}_2^2 + \norm{\Delta\evec}_2^2 + 2\langle\evec, \Delta\evec\rangle \le \norm{\evec}_2^2 + \norm{\Delta\evec}_2^2.
    \] 
We show $\langle\evec,\Delta\evec\rangle\le0$ by the same comparison trick as in the proof of \Cref{lem:weakavg}, now with the frozen errors as weights and the current fractional orientation as the starting point. For a processed pair with $e_u\le e_v$, rounding transfers some amount $r\ge0$ of outgoing degree from $v$ to $u$. Its contribution to the inner product is $r(e_u-e_v)\le0$. Summing over the processed pairs gives $\langle\evec,\Delta\evec\rangle\le0$.

Let $m_{ij}$ be the number of original edges with an endpoint in $V_i$ and one in $V_j$ (without taking their direction into account). Clearly, the error at a vertex $u$ cannot exceed the number of edges incident to $u$ and $V_j$. Moreover, since each block contains at most  $b$ vertices, this number of edges cannot exceed $2b$. Hence, for any vertex $u \in V_i$ (analogously for $u \in V_j$), it holds that
\[
    |\Delta e(u)|^2 \le 2b \sum_{\substack{v \neq u \\ v \in V_j}} M_{uv}.
\] 
By adding up the inequality over all vertices of $V_i$ and $V_j$, we obtain
\[
    \norm{\Delta\evec}_2^2 \le 4b \sum_{\substack{u \neq v \\ u \in V_i , v \in V_j}} M_{uv} = 4bm_{ij}.
\] 
Note that for blocks $V_i$ and $V_j$, the inequality is added twice, hence the additional factor of two.

By adding up the inequalities for every pair of blocks, we finally obtain that
\[
    \norm{\outvec_{\yvec} - \dvec}_2^2 \le \norm{\outvec_{\xvec} - \dvec}_2^2 + \sum_{i \le j} 4bm_{ij} = \norm{\outvec_{\xvec} - \dvec}_2^2 + 4bm. \qedhere
\] 

Intuitively, \Cref{lem:block} tells us that $b$ must not be too large, as otherwise, the error would grow too quickly. On the other hand, it can also not be too small, as otherwise, the rounding algorithm would essentially be able to also reconstruct the entire graph.

\end{proof}

\begin{corollary}\label{cor:rounding-tradeoff}
For any integer $1\le b\le n$, averaging for $T=\lceil n/b\rceil$ rounds followed by greedy block rounding computes a known pseudo-orientation $\yvec$ with integral degrees and its exact degree-error vector $\evec=\outvec_{\yvec}-\dvec$, using
\[
 \cO\!\left(\frac{n^2}{b}\log\left(1+\frac nb\right)\right)
\]
cut queries and satisfying $\norm{\evec}_1=\cO(n^{3/2}\sqrt b)$.
\end{corollary}
\begin{proof}
By \Cref{lem:weakavg,lem:block}, and since $m\le n^2$,
\[
 \norm{\evec}_2^2\le\frac{4n^3}{T}+4bm=\cO(n^2b).
\]
Cauchy--Schwarz gives $\norm{\evec}_1\le\sqrt n\,\norm{\evec}_2=\cO(n^{3/2}\sqrt b)$. Averaging uses $\cO(nT)$ queries, and \Cref{lem:block-queries} bounds the rounding cost by $\cO((n^2/b)\log(1 + T))$. Substituting $T=\lceil n/b\rceil$ gives the claimed query bound.
\end{proof}

\section{Deterministic Repair by Shortest Augmenting Paths} \label{sec:repair}

In this section, we iteratively change $\yvec$ until the out-degree $\outvec_{\yvec}$ matches the true out-degree $\dvec$ at each vertex. To decrease the $\ell_1$ error by two, we find a directed path from a vertex with too much out-degree under $\yvec$ to a vertex with too little out-degree and we flip the directions of all edges on this path.

To find this collection of paths, we will use Dinic's algorithm over a suitably defined auxiliary graph. We then adapt it to run in the cut query model over the given pseudo-orientation $\yvec$. A textbook analysis shows that Dinic's algorithm runs in few phases and modifies only a small number of edges. Using these facts, we show that it can be adapted to use few cut queries as well.

\newcommand{\networky}{\mathcal{N}(\yvec)}

Formally, let $\yvec$ be a pseudo-orientation with integral degrees. Recall that $\evec = \outvec_{\yvec} - \dvec$. Our objective is to modify $\yvec$ so that $\evec$ becomes zero. To do so, let $\networky$ be the directed flow network over vertices $V \cup \{s, t\}$ with the following properties:

\begin{itemize}
    \item vertices $s$ and $t$ are the source and the sink respectively; they are artificially added to the original graph;
    \item edges  $s \to u$ are added with capacity $e_u$ for all vertices $u$ with $e_u > 0$; similarly, edges $u \to t$ are added with capacity $-e_u$ for all vertices $u$ with $e_u < 0$;
    \item for every pair $u, v \in V$, add edge $u \to v$ with capacity $M_{uv}y_{uv}$; these are the edges of the graph skeleton directed according to the pseudo-orientation $\yvec$.
\end{itemize}

The (residual) network is maintained implicitly through $\evec$ and $\yvec$. The algorithm has access to edges incident to the source and sink, but does not know the edges inside $V^2$. To operate on the internal edges, we will use cut queries instead to discover them and to access their capacities. The main advantage of using the pseudo-orientation here is that we can now distinguish between the in-neighbors and the out-neighbors of some vertex.
The following lemma shows how the relevant subroutines from \Cref{section:subroutines} can be adapted to pseudo-orientations.

\begin{lemma}[Discovering an Edge Under a Pseudo-Orientation] \label{lem:pseudo-discovery}
    Let $\yvec$ be a pseudo-orientation with integral degrees, $u \in V$ and $X \subseteq V \setminus \{u\}$. The cut from $u$ to $X$ and the cut from $X$ to $u$ under pseudo-orientation $\yvec$ can both be computed using $\mathcal{O}(1)$ cut queries. If either cut is positive, an edge in the respective direction can be found using $\mathcal{O}(\log n)$ queries.
\end{lemma}

\begin{proof}
    Let $X_\alpha = \{v \in X : y_{uv} = \alpha\}$. In other words, $X_1$ contains all outgoing neighbors of $u$ such that all edges are directed out of $u$ and $X_{\nicefrac 12}$ contains all neighbors with both one outgoing and one incoming edge. Then, the cut from $u$ to $X$ under $\yvec$ is:
    \[
        \sum_{v \in X} M_{uv} y_{uv} = \cross(u, X_1) + \frac{1}{2}\cross(u, X_{\nicefrac 12}).
    \] 
    \noindent Similarly, for the cut from $X$ to $u$, we define $X'_\alpha = \{v \in X : y_{vu} = \alpha\}$. By \Cref{lemma:primitives}, the above formula uses $\mathcal{O}(1)$ cut queries. Note that $\cross$ does not take orientation into account, therefore, in the case of $M_{uv} = 2$ and $y_{uv} = \nicefrac 12$, we make sure that we do not count both edges, hence why the distinction between $X_1$ and $X_{\nicefrac 12}$ is required.

    The binary search from \Cref{lemma:discover-edge} can then be used with this directed cut subroutine instead of $\cross$ to discover an edge from $u$ to $X$ and vice versa.
\end{proof}

Suppose that we find a directed path from $s$ to $t$ in the residual network along which we want to push one unit of flow. In the residual network, this can be seen as flipping the edges on said path. The following lemma shows how to maintain $\yvec$ and $\evec$ under the operation of pushing flow.

\begin{lemma}[Flipping a Path in a Pseudo-Orientation]
    Let $\yvec$ be a pseudo-orientation with integral degrees. A path $P$ can be reversed under $\yvec$ using $\mathcal{O}(|P|)$ cut queries. Doing so affects the error $\evec$ only in the endpoints.
\end{lemma}

\begin{figure}[tbp]
\centering
\begin{tikzpicture}[x=1cm,y=1cm]
  \node[vu title] at (2.4,1.2) {Before Reversal};
  \node[vu title] at (10,1.2) {After Reversal};
  \foreach \name/\xx/\err/\errcolor in {u/0/+1/VUOrange,a/1.6/0/black,b/3.2/0/black,v/4.8/-1/VUBlue}{
    \node[vu vertex,draw=\errcolor] (L\name) at (\xx,0) {$\name$};
    \node[vu text,text=\errcolor] at (\xx,-.65) {$e_{\name}=\err$};
    \node[vu vertex] (R\name) at (\xx+7.6,0) {$\name$};
    \node[vu text] at (\xx+7.6,-.65) {$e_{\name}=0$};
  }
  \draw[vu arc] (Lu)--(La); \draw[vu arc] (La)--(Lb); \draw[vu arc] (Lb)--(Lv);
  \draw[vu arc] (Ra)--(Ru); \draw[vu arc] (Rb)--(Ra); \draw[vu arc] (Rv)--(Rb);
  \draw[vu arc,black!55] (5.5,0)--node[above,font=\small] {reverse} (6.9,0);
  \node[vu text,text=black] at (2.4,-1.25) {$\norm{\evec}_1=2$};
  \node[vu text,text=black] at (10,-1.25) {$\norm{\evec}_1=0$};
\end{tikzpicture}
\caption{Reversing a path from positive to negative degree error changes only the endpoint degrees and reduces the $\ell_1$ error by two.}
\label{fig:repair}
\end{figure}

\begin{proof}
    To reverse an edge copy $u \to v$, query $M_{uv} = \cross(u, v)$ and set
    \[
        y_{uv} \gets y_{uv} - \frac{1}{M_{uv}}, \qquad
        y_{vu} \gets y_{vu} + \frac{1}{M_{uv}}
    \]
    for each edge on the path. Subtracting $1 / M_{uv}$ ensures that the induced edge counts change by minus one and plus one, respectively for each edge; therefore, they remain integral. Along a path, the changes at internal vertices cancel; therefore the error stays the same, while the error decreases by one and increases by one for the first vertex and the last vertex, respectively.

    This procedure uses $\mathcal{O}(|P|)$ cut queries in total.
\end{proof}

Finally, the following lemma shows that by applying any maximum flow algorithm, we obtain a pseudo-orientation with matching out-degrees. Equivalently, the $\ell_1$ degree error is reduced to zero.

\begin{lemma}
    Let $\yvec$ be a pseudo-orientation with integral degrees. The maximum flow in  $\networky$ is $\norm{\evec}_1 / 2$.
\end{lemma}

\begin{proof}
    Since both out-degree vectors $\outvec_{\yvec}$ and $\dvec$ have coordinate sum $m$, it holds that
    \[
        \sum_{u} (e_u)_{+} = \sum_{u} (-e_u)_{+} = \frac{1}{2} \norm{\evec}_1.
    \] 

    We prove that at each step of the maxflow algorithm, when $\norm{\evec} > 0$, there is an augmenting path \ie{a path from a node with positive error to a node with negative error}.

    Let $S$ be the set of all vertices reachable from all vertices with positive error. Assume towards a contradiction that $S$ contains no vertex with negative error.

    By the identity of \Cref{lemma:primitives}, we have that
    \[
        \query(S) = \left( \sum_{u \in S} \outdegree_u \right) - \internal(S), \qquad
        \query_{\yvec}(S) = \left( \sum_{u \in S} \out_{\yvec}(u) \right) - \internal(S)
    \] 
    where $\query_{\yvec}(S)$ is the size of the directed cut under pseudo-orientation $\yvec$. By assumption, we then have that $\query_{\yvec}(S) > \query(S) \ge 0$. Then, there must be a node outside of $S$ which is reachable from $S$, which contradicts the definition of $S$. Each unit augmentation reduces the $\ell_1$ error by two, so a flow of value half the initial error makes the error zero and saturates the source capacity.
\end{proof}

\subsection{Dinic's Algorithm Under Pseudo-Orientations}
We assume familiarity with this textbook algorithm, but have made an effort to make the ideas in this section lucid without prior knowledge (though we have decided not to reprove well-known properties of the algorithm). 
A good first exposition can be found in \cite{erickson2019maxflow}.
See \cite[Section 8.2]{ahujaMagnantiOrlin1993} for a full formal treatment.
The algorithm operates in phases based on the distance between $s$ and $t$: each time that distance increases a new phase starts.

\paragraph{Layer Graph.} To perform a BFS, we enqueue all out-neighbors of source vertex $s$. When processing a vertex $v$ from the original graph, we use \Cref{lem:pseudo-discovery} to find all the edges from $v$ to its \textit{unvisited} out-neighbors given by the pseudo-orientation $\yvec$. 

Let $L_d$ be the set of all vertices that lie at distance exactly $d$ from the source and let $\delta$ be the length of a shortest path from $s$ to $t$. Ignore all vertices that lie at distance $\delta$ or greater except for the sink $t$. Then, the layer graph is given by all the vertices of $L_0, \ldots, L_\delta$ and the directed edges from $L_i$ to $L_{i + 1}$ for any $i$ with $0 \le i < \delta$.

Finding the layer sets takes $\mathcal{O}(n \log n)$ cut queries.

\paragraph{Finding Augmenting Paths.} During a phase, we repeatedly perform DFS searches on the layer graph to find a blocking flow.

A vertex is called \textit{dead} once it is proved to have no path to $t$ in the layer graph; otherwise, it is called \textit{alive}. During a DFS at vertex $v$, we use \Cref{lem:pseudo-discovery} to find an edge towards an \textit{alive} vertex on the next layer that we then recursively explore. If, at any point, the DFS arrives at the sink, we can push a unit of flow on the augmenting path and restart from $s$. On the other hand, if vertex $v$  has no more \textit{alive} out-neighbors, it stops its search and $v$ is marked as \textit{dead}. Layers stay fixed during a phase. New reverse arcs point backward, so dead vertices remain dead.

Let $f_\delta$ be the total amount of flow pushed during this phase.
Clearly, for each unit of flow pushed, we explore $\delta$ edges (at a cost of $\cO(\delta \log n)$ queries).
Additionally, we also explore edges towards \textit{alive} vertices that prove to be \textit{dead}.
Since each unsuccessful advance is charged to its endpoint becoming dead, and each vertex becomes dead at most once, finding a blocking flow of size $f_\delta$ on the layer graph with $\delta$ layers uses $\mathcal{O}(((\delta - 2) f_\delta + n) \log n)$ cut queries.

Note that since $s$ and $t$ are two artificially added nodes, we do not need to use cut queries to discover edges incident to them; hence, on a path of length $\delta$, we use cut queries only for the $\delta - 2$ non-boundary edges.

\paragraph{Query Analysis.} Recall the following bounds on the length of an augmenting path and the number of phases which are the results of a textbook analysis of Dinic's algorithm:

\begin{lemma}[Layer Bound {\cite[Section~8.2, p.~253]{ahujaMagnantiOrlin1993}}] \label{lemma:layerbound}
    Let $f>0$ be the remaining flow to be pushed on the entire residual graph. It holds that
    \[
    \delta - 2 \le \frac{4n}{\sqrt{f}}.
    \] 
\end{lemma}

\begin{lemma}[{\cite[Sections~7.5 and~8.2, pp.~222 and~253]{ahujaMagnantiOrlin1993}}]
    Dinic's algorithm requires $\mathcal{O}(n^{2 / 3})$ phases.
\end{lemma}

The textbook layer-cut argument applies here using only cuts between internal BFS layers, so the possibly large capacities incident to $s$ and $t$ do not enter the bound. Residual capacities between each internal vertex pair are at most two, which changes only the constant. Each blocking-flow phase still strictly increases the shortest-path distance. After $\lceil n^{2/3}\rceil$ phases, the layer bound leaves only $\mathcal{O}(n^{2/3})$ units of flow, and hence at most that many further phases, since every phase pushes at least one unit.

Combining the above two lemmas gives the cut query bound of Dinic's algorithm.

\begin{lemma}[Repair] \label{lem:repair}
    Given a known pseudo-orientation with integral degrees and initial $\ell_1$ degree error $F$, we can match its out-degrees to $\dvec$ using
    \[
        \mathcal{O}((n^{5 / 3} + n\sqrt{F}) \log n)
    \]
    cut queries, deterministically. In particular, an error of $F = \mathcal{O}(n^{8 / 5})$ can be repaired using $\mathcal{O}(n^{9 / 5} \log n)$ queries.
\end{lemma}

\begin{proof}
    For each phase, the algorithm incurs $\mathcal{O}(n \log n)$ cut queries for computing the layer graph and for the binary searches that turn vertices \textit{dead}, totalling $\mathcal{O}(n^{5 / 3} \log n)$ cut queries. It remains to bound the rest of the cut queries that are used to discover augmenting paths.

    We apply \Cref{lemma:layerbound} over each augmenting path found. Paths remain shortest throughout a phase, since residual updates cannot decrease distances. Put $R=F/2$, the number of unit augmentations. Let $P_i$ be the $i$-th augmenting path found by the algorithm. For the $i$-th augmenting path, $R - i + 1$ units of flow remain and the algorithm incurs a cost of $\mathcal{O}((|P_i| - 2) \log n)$ cut queries. By summing up all the bounds, we obtain
    \[
        \sum_{i=1}^{R} (|P_i|-2) \le \sum_{i=1}^{R} \frac{4n}{\sqrt{R-i+1}}
        =4n\sum_{i=1}^{R}\frac{1}{\sqrt{i}} \le 8n\sqrt{R}
        =\mathcal{O}(n\sqrt{F}).
    \]

    Adding the two cut query bounds, we obtain that Dinic's algorithm uses $\mathcal{O}\left( (n^{5 / 3} + n\sqrt{F}) \log n \right)$ cut queries to repair the pseudo-orientation.
\end{proof}

\section{Deriving the Reachability, SCC, and Minimum $s$-$t$ Cut Algorithms}\label{sec:consequences}
The pseudo-orientation with the same out-degrees as $G$ is surprisingly powerful, solving not only $s$-$t$ reachability as we set out to do.
In this section, we discuss how to use it to solve reachability and other problems.
But first, let us explicitly put together all the parts and balance the parameters to establish the following result.

\alggeneralorientation

\begin{proof}
  Set $b=\lceil n^{1/5}\rceil$. By \Cref{cor:rounding-tradeoff}, we obtain a known pseudo-orientation with integral degrees and exact degree-error vector of $\ell_1$ norm $F=\cO(n^{8/5})$ using $\cO(n^{9/5}\log n)$ queries. By \Cref{lem:repair}, repairing it to have the correct out-degrees uses $\cO(n^{9/5}\log n)$ further queries.
\end{proof}

\subsection{Single-Source Reachability}
Let $\yvec$ be the known pseudo-orientation with integral degrees and $\outvec_{\yvec}=\dvec$, and let $H$ be its induced skeleton orientation.
We can search $H$ without learning all its edges: when processing a vertex $u$, we only need to find edges to vertices that have not yet been discovered.
By \Cref{lem:pseudo-discovery}, we can test whether such an edge exists using $\cO(1)$ cut queries and, if it does, find one using $\cO(\log n)$ queries.
Once a vertex is discovered, we exclude it from all subsequent neighbor searches.
Thus even a dense graph can be explored by discovering only the edges of a search tree.

A positive oracle answer $\query(S)$ certifies that some edge leaves $S$, but does not in general let us identify that directed edge in $G$.
Knowing $\yvec$ supplies the direction choices needed by \Cref{lem:pseudo-discovery}, so we can identify outgoing edges in $H$.
These edges may differ from those of $G$; the lemma below shows that the reachable vertices are nevertheless the same.
\Cref{alg:pseudo-reachability} gives the search explicitly.

\begin{algorithm}[H]
    \begin{algorithmic}[1]
    \item[]\textbf{Input:} Cut oracle $\query$, source $s$, and known pseudo-orientation $\yvec$ as above.
    \item[]\textbf{Output:} Set $R$ of vertices reachable from $s$.

    \State $R\gets\{s\}$, $U\gets V\setminus\{s\}$; initialize a queue $Q$ containing $s$.
    \While{$Q\ne\emptyset$}
        \State Remove a vertex $u$ from $Q$.
        \While{the cut from $u$ to $U$ under $\yvec$ is positive}
            \State Discover $v\in U$ with $u\to v$ in $H$ using \Cref{lem:pseudo-discovery}.
            \State $R\gets R\cup\{v\}$, $U\gets U\setminus\{v\}$; append $v$ to $Q$.
        \EndWhile
    \EndWhile
    \State \textbf{return} $R$
    \end{algorithmic}
    \caption{Single-source reachability using a known pseudo-orientation.}
    \label{alg:pseudo-reachability}
\end{algorithm}

Each successful neighbor search discovers a new vertex, and each processed vertex requires one final test for a zero cut.
There are at most $n-1$ discoveries and $n$ such final tests, so the total query cost is $\cO(n\log n)$.
Every vertex added to $R$ is reachable from $s$ in $H$.
At termination, no edge of $H$ leaves $R$, so $R$ contains every vertex reachable from $s$ in $H$.
To transfer this conclusion to $G$, we use the following observation.

\begin{lemma}[Preserving Cuts and Reachability]\label{lem:cuts}
The graphs $G$ and $H$ have the same directed cut function and reachability relation.
\end{lemma}
\begin{proof}
Reorienting the skeleton leaves the number $\internal(S)$ of copies internal to any set $S$ unchanged. Since the out-degrees match, \Cref{lemma:primitives} gives
\[
 \cut_H(S)=\sum_{u\in S}\out_{\yvec}(u)-\internal(S)
 =\sum_{u\in S}\outdegree(u)-\internal(S)=\query(S).
\]
A vertex $t$ is unreachable from $s$ exactly when some set containing $s$ and excluding $t$ has no outgoing arc. Equality of the cut functions therefore implies equality of reachability.
\end{proof}

This proves the following.
\begin{theorem}[Reachability]\label{thm:orientation-reachability}
Given a known pseudo-orientation with integral degrees and out-degree vector $\dvec$, single-source reachability can be decided using $\cO(n\log n)$ additional cut queries.
\end{theorem}

Including the construction of the pseudo-orientation from \Cref{thm:pseudo-orientation} gives the following result.

\alggeneralreachability

Note that this implies that we can solve multi-source reachability with a source set of size $\cO(n^{1-\smol})$ in $\cO(n^{2-\smol} \log n)$ queries for any $\smol \in [0, 1/5]$.
In the edge case this is close to optimal as all-pairs reachability requires $\Omega(n^2 / \log n)$ queries in \Cref{rem:skeleton-information}.
It is an interesting question whether this trade-off can be improved for $\smol > 0$.

\subsection{Strongly Connected Components}
We can find the SCCs of $G$, in topological order, by running Kosaraju's algorithm.
Recall that Kosaraju's algorithm essentially consists of a forward graph search in $G$ (as above) and a backward graph search (equivalently, a forward search in the graph with all edge directions flipped).

Searching backward from $s$ means following incoming edges in $H$, thereby finding the vertices that can reach $s$.
We already know a pseudo-orientation for the transpose graph $H^{\mathrm{rev}}$: set $y^{\mathrm{rev}}_{uv}=y_{vu}=1-y_{uv}$.
The skeleton stays the same, and the induced edge counts remain integral.
Concretely, for a vertex $u$ and a set $X\subseteq V\setminus\{u\}$, the number of edges from $u$ to $X$ in $H^{\mathrm{rev}}$ is
\[
 \sum_{v\in X}M_{uv}y_{vu}
 =\cross(u,X)-\sum_{v\in X}M_{uv}y_{uv}.
\]
The first term counts all skeleton copies between $u$ and $X$, and the second counts those directed out of $u$ in $H$; the remainder counts those directed into $u$.
Both terms use $\cO(1)$ queries to the original oracle by \Cref{lem:pseudo-discovery,lemma:primitives}.
If the remainder is positive, the same binary search finds an in-neighbor using $\cO(\log n)$ queries.
Thus the backward search has the same $\cO(n\log n)$ query cost and uses the pseudo-orientation we have already constructed.

Implementing Kosaraju's SCC algorithm with our forward and backward search primitives gives the following theorem.

\alggeneralscc

\subsection{Minimum Cuts and Maximum Flow}
Interestingly, in our effort to produce a reachability algorithm, we have essentially already shown how to compute a minimum $s$-$t$ cut in \Cref{sec:repair}.
We can use the implementation of Dinic's algorithm from that section to compute the maximum flow between any two vertices in $G$ (instead of the slightly modified auxiliary graph used for repair).

\alggeneralmincut
\begin{proof}
    Apply Dinic' algorithm as described in \Cref{sec:repair}, the only difference being that $s$ and $t$ are not artificially added to the graph anymore. The same analysis of the algorithm still holds.
Residual pair capacities still sum to $M_{uv}\le2$, so the same layer bound and accounting apply; the maximum flow value is at most $\outdegree(s)\le n-1$.
Thus the algorithm uses $\cO(n^{5/3}\log n)$ additional queries.
At termination, the vertices reachable from $s$ in the residual graph form a minimum $s$--$t$ cut in $H$.
By \Cref{lem:cuts}, this is also a minimum $s$--$t$ cut in $G$ with the same value.
The max-flow min-cut theorem identifies this value with the maximum flow value and the maximum number of edge-disjoint $s$--$t$ paths in $G$.
Including the construction of $\yvec$, these tasks use $\cO(n^{9/5}\log n)$ queries.
\end{proof}

\section{Open Questions}
Of course, the most natural question is to improve the number of queries required to find a pseudo-orientation with correct out-degrees.
Similarly, establishing a lower bound would show the limits of this technique (and may be easier than showing a lower bound on $s$-$t$ reachability in general).

For undirected graphs, we know of subquadratic algorithms for global minimum cut~\cite{rubinstein_et_al:LIPIcs.ITCS.2018.39}.
Finding one for directed graphs (through a pseudo-orientation or not) remains open.

\section*{Acknowledgments} We would like to thank Danupon Nanongkai for posting the open problem on cut-query complexity of $s$-$t$ reachability in general graphs \cite{nanongkai-open-problems}, and for informing us of other concurrent work.

This work was supported in part by a grant of access to OpenAI models through the ChatGPT for Academic Researchers program.

\section*{AI Use Declaration}
The DAG part (\Cref{section:topsort,section:applications,sec:toposort-scc}) of this paper was reasoned and written entirely without the use of AI tools.
The algorithm for general directed graphs (\Cref{sec:orientations,sec:det,sec:repair,sec:consequences}) was discovered with help from OpenAI's ChatGPT 6.1.
The authors carefully checked, significantly simplified, rewrote and improved the clarity of these sections.

\clearpage
\printbibliography
\clearpage

\appendix
\crefalias{section}{appendix}
\section{Topological Sorting of Strongly Connected Components}
\label{sec:toposort-scc}

\newcommand{\CC}{\mathcal{C}}
\newcommand{\SetS}{\mathcal{S}}
\newcommand{\SCC}{\mathrm{SCC}}

In this section, we prove the following result which generalizes \Cref{thm:algtopsort}.

\begin{theorem}\label{thm:algtopsortscc}
     Let $G = (V, E)$ be a directed graph and $\mathcal{C} = (C_1, C_2, \ldots, C_k)$ be a partition of the vertices into strongly connected components. There is an algorithm that computes the topological sort of the strongly connected components of $G$ using $\cO(n \sqrt{n \log n})$ cut queries.
\end{theorem}

When we attempt this generalization to all directed graphs, complications appear due to the existence of cycles.
As stated before, the directions of the edges on a directed cycle cannot be determined.
Functions that relied on the structure of DAGs do not offer the same guarantees anymore.
Moreover, since every vertex may now have a positive indegree, \Cref{alg:topsort} does not work anymore.

In this section, we show that if we already know the strongly connected components of our graph, we can then topologically sort them, which would lead to a reachability algorithm for general graphs.
Briefly, by condensing a general directed graph into its strongly connected components, we obtain a DAG, which can then be topologically sorted. Since previously used subroutines and properties still hold on the condensation graph, we can then reduce the single-source reachability problem on the original graph $G$ to deciding single-source reachability starting from $\SCC(s)$ in the condensed graph.
This suggests that finding the strongly connected components is the core challenge in finding a reachability algorithm for general directed graphs.
In particular, it raises the interesting intermediate open problem of detecting whether a graph is strongly connected using polynomially fewer than $n^2$ queries.

Since the framework used to prove \Cref{thm:algtopsortscc} is similar to that of \Cref{thm:algtopsort} and the proofs flow nearly identically, we will only highlight the differences.
Assume that graph $G$ has $k$ strongly connected components. We denote the partition of vertices into strongly connected components by $\CC = \{C_1, C_2, \ldots, C_k\}$. A topological sort of SCCs $\SetS = [S_1, S_2, \ldots, S_k]$ is defined as a permutation of $\CC$. Similarly, we define the prefix of length $i$ of a sequence as $\SetS_i = [S_1, S_2, \ldots, S_i]$. For any sequence of SCCs $\SetS$, we write $V(\SetS)$ for all vertices in any component in $\SetS$.
The analogous version of \Cref{prop:remainingindegree} holds for the SCC case by the same argument.

\begin{observation}[Remaining Indegree of an SCC] \label{prop:sccremindegree}
    Let $\SetS'$ be the prefix of a topological sort of SCCs and $C \in \CC \setminus \SetS'$. The indegree of $C$ in $G - V(\SetS')$ is $\querycomp(C) - \cross(V(\SetS'), C)$.
\end{observation}

We now adapt \Cref{obs:indegreehops} for SCCs. Similarly, if an SCC has large indegree, we want to delay it for long enough. An issue is that all the vertices of an SCC may share in-neighbors. Fortunately, we can take the average of the indegree. Then, there must be a vertex in the SCC which must wait for enough in-neighbors to be placed in the topological sort.

\begin{observation} \label{obs:indegreehops-scc}
    Let $G = (V, E)$ be a directed graph and $C \in \CC$ an SCC of $G$. Then, for any topological sort over SCCs $\SetS = [S_1, S_2, \ldots, S_{k}]$ and integer $i$ such that $C \not\in \SetS_i$, then it holds that $C \not\in \SetS_j$ for any $j$ such that $|V(\SetS_j)| \le |V(\SetS_i)| + \left\lceil   \frac{\remindegree}{|C|} \right\rceil$ where $\remindegree := \querycomp(C) - \cross(\SetS_i, C)$.
\end{observation}

\begin{proof}
    Similarly, $\remindegree$ is the indegree of $C$ in $G - \SetS'_i$. Since the number of incoming edges between $G - \SetS'_i$ and $C$ is additive over vertices of $C$, then there must be a vertex with at least $\left\lceil \frac{\remindegree}{|C|} \right\rceil$ incoming edges from outside of $C$.
    Since $C$ must go after at least $\remindegree$ in-neighbors of $C$, then $C$ cannot be part of the topological sort prefix of size $|V(\SetS_i)| + \remindegree$.
\end{proof}

Using the above, we obtain \Cref{alg:topsort-scc}. Conceptually, the only difference between this algorithm and its DAG counterpart is that at each point, we average the remaining indegree to the size of the SCC considered. An SCC is \textit{heavy} if its \textit{remaining average indegree} is larger than threshold $B$ and similarly, a \textit{light} SCC has an average indegree less than or equal to $B$. Otherwise, the flow of the algorithm is exactly the same as with \Cref{alg:topsort}.

\begin{algorithm}
    \begin{algorithmic}[1]
        \item[]\textbf{Input:} Cut query oracle function $\query$, heaviness threshold $B$ and partition of vertices into SCCs.
        \item[]\textbf{Output:} Topological ordering $\SetS'$.

        \item[]

        \State Compute $\indegree(v)$ and $\outdegree(v)$ to enable \Cref{lemma:primitives}.
        \State Mark all SCCs as \textit{heavy}.

        \Procedure{Append}{U}
            \State Append $U$ to $\SetS'$.
            \State Discover all edges between $U$ and light SCCs using \Cref{lemma:discover-edge}.
            \For{each discovered edge $u \to v$}
                \State Let $\remindegree$ be the remaining indegree of $\SCC(v)$. 
                \If{$\remindegree = 0$}
                    \Call{Append}{$\SCC(v)$}

                \EndIf             
            \EndFor
        \EndProcedure

        \item[]

        \While{$|\SetS'| < n$}
            \For{each \textit{heavy} SCC $U$}
                \State Let $\remindegree$ be the remaining indegree of $U$.
                \If{$\remindegree = 0$} \Call{Append}{$U$}
                \ElsIf{$\frac{\remindegree}{|U|} \le B$} Mark $U$ as \textit{light}.
                \Else{} Avoid checking $U$ until $\setsize{V(\SetS')}$ increases by at least $\left\lceil \frac{\remindegree}{|U|} \right\rceil$.
                \EndIf
            \EndFor
        \EndWhile

    \end{algorithmic}
    \caption{Adapted algorithm to compute a topological sort of SCCs in the cut query model.}
    \label{alg:topsort-scc}
\end{algorithm}

\begin{lemma} \label{lemma:topsortprefix-scc}
    Let $G = (V, E)$ be a directed graph and $\CC$ be a partition of $V$ into SCCs. Consider the sequence $\SetS'_{k'}$ defined in \Cref{alg:topsort-scc} at any point in time. Then, $\SetS'_{k'}$ is the prefix of a topological sort of the SCCs of $G$.
\end{lemma}

\begin{lemma} \label{lemma:correctness-scc}
    \Cref{alg:topsort-scc} terminates and outputs a topological sort of the SCCs of directed graph $G$.
\end{lemma}

We now select parameter $B$ to obtain the main result of this section. NB, \Cref{thm:algtopsortscc} assumes that the algorithm receives the partition of vertices into SCCs.

\begin{proof}[Proof of \Cref{thm:algtopsortscc}]
    The analysis is similar to the proof of \Cref{thm:algtopsort}, relying on the analogous \Cref{obs:indegreehops-scc,lemma:correctness-scc,lemma:topsortprefix-scc}, with some differences that will be highlighted here.

    To bound the queries at Lines~5-7, we again count the number of edges that are discovered. Consider the moment when an SCC $U$ is marked as \textit{light}. Again, $U$ will have $\remindegree$ incoming edges with $\remindegree$ computed at Line~11. Since $U$ is added if $\frac{\remindegree}{|U|} \le B$, we get that $U$ has at most $B|U|$ incoming edges. By summing up over all SCCs, we get that we discover at most $B \sum_{i=1}^{k} |C_i| = Bn$ edges. Therefore, the number of queries at Lines~5-7 is $\cO(nB \log n)$.

    To bound the queries at Line~11, we again consider for an individual SCC $U$ the moments it is iterated through at Line~10. Let $t$ be the number of iterations through $U$ and let $n'_1, n'_2, \ldots, n'_t$ be the sizes of $\SetS'$ at each of those moments. We define $l_i = |V(\SetS'_{n'_i})|$ for every  $i \in [t]$. We again bound the difference of adjacent elements: $l_{i + 1} - l_{i} = \left\lceil \frac{\remindegree}{|U|} \right\rceil > B$. By adding up all the inequalities, we obtain $n \ge l_{t} - l_1 > (t - 1)B$. Therefore, we obtain that an SCC is checked $\cO(n / B)$ times. By summing over all SCCs, we can bound the number of queries at Line~11 by $\cO(nk / B) = \cO(n^2 / B)$.

    By picking $B = \sqrt{n / \log n} $, we obtain that \Cref{alg:topsort-scc} uses $\cO(n \sqrt{n \log n})$ queries.
\end{proof}

\end{document}